\documentclass[letterpaper,twocolumn,10pt]{article}

\usepackage{algorithmicx}
\usepackage{algorithm}
\usepackage{algpseudocode}
\usepackage{amsmath}
\usepackage{bm}
\usepackage{booktabs}
\usepackage{calc}
\usepackage{lipsum}
\usepackage{pervasives}
\usepackage{subcaption}
\usepackage{tikz}
\usepackage{usenix2019_v3}
\usepackage{wrapfig}
\usetikzlibrary{calc}
\usetikzlibrary{positioning}

\newcommand{\deps}[1]{\text{deps}}
\newcommand{\msg}[1]{\langle #1 \rangle}
\newcommand{\leadercolor}{flatblue}
\newcommand{\proposercolor}{flatpurple}
\newcommand{\replicacolor}{flatred}
\newcommand{\clientcolor}{white}
\newcommand{\depservicecolor}{flatgreenalt}
\newcommand{\consensuscolor}{flatbrown}

\algtext*{EndIf}
\algtext*{EndFor}

\algblockdefx[Upon]{Upon}{EndUpon}
  [1]{\textbf{upon} #1 \textbf{do}}
  {end}
\algtext*{EndUpon}

\algnewcommand{\algorithmicstate}{\textbf{State:}}
\algnewcommand{\GlobalState}{\item[\algorithmicstate]}

\begin{document}
\date{}

\title{Bipartisan Paxos: A Modular State Machine Replication Protocol}

\author{%
  \rm
  Michael Whittaker$^1$,
  Neil Giridharan$^1$,
  Adriana Szekeres$^2$,
  Joseph M. Hellerstein$^1$,
  Ion Stoica$^1$\\
  $^1$UC Berkeley,
  $^2$University of Washington
}

\maketitle

{\begin{abstract}
  There is no shortage of state machine replication protocols. From Generalized
  Paxos to EPaxos, a huge number of replication protocols have been proposed
  that achieve high throughput \emph{and} low latency. However, these protocols
  all have two problems. First, they do not scale. Many protocols actually slow
  down when you scale them, instead of speeding up. For example, increasing the
  number of MultiPaxos acceptors increases quorum sizes and slows down the
  protocol. Second, they are too complicated. This is not a secret; state
  machine replication is notoriously difficult to understand.

  In this paper, we tackle both problems with a single solution:
  \emph{modularity}. We present Bipartisan Paxos (BPaxos), a modular state
  machine replication protocol. Modularity yields high throughput via scaling.
  We note that while many replication protocol components do not scale, some
  do. By modularizing BPaxos, we are able to disentangle the two and scale the
  bottleneck components to increase the protocol's throughput. Modularity also
  yields simplicity. BPaxos is divided into a number of independent modules
  that can be understood and proven correct in isolation.
\end{abstract}
}
{\section{Introduction}
State machine replication protocols like MultiPaxos~\cite{lamport1998part,
lamport2001paxos} and Raft~\cite{ongaro2014search} allow a state machine to be
executed in unison across a number of machines, despite the possibility of
faults. Today, state machine replication is pervasive. Nearly every strongly
consistent distributed system is implemented with some form of state machine
replication~\cite{corbett2013spanner, thomson2012calvin, burrows2006chubby,
baker2011megastore, cockroach2019website, cosmos2019website, tidb2019website,
yugabyte2019website}.

MultiPaxos is one of the oldest and one of the most widely used state machine
replication protocols. However, despite its popularity, MultiPaxos does not have
optimal throughput or optimal latency. In response, a number of state machine
replication protocols have been proposed to address MultiPaxos' suboptimal
performance~\cite{%
  arun2017speeding,
  biely2012s,
  howard2016flexible,
  lamport2005generalized,
  lamport2006fast,
  li2016just,
  mao2008mencius,
  moraru2013there,
  nawab2018dpaxos,
  park2019exploiting,
  ports2015designing
}.
These protocols use sophisticated techniques that either increase MultiPaxos'
throughput, decrease its latency, or both

These sophisticated replication protocols have two shortcomings: they do not
scale, and they are very complex. In this paper, we address both these
shortcomings with a single solution: \textbf{modularity}. We present Bipartisan
Paxos (BPaxos), a state machine replication protocol that is composed of a
number of independent modules. Modularity allows us to achieve state-of-the-art
throughput via a straightforward form of scaling. Furthermore, modularity makes
BPaxos significantly easier to understand compared to similar protocols.

\paragraph{Scaling}
Simple state machine replication protocols like MultiPaxos and Raft cannot take
advantage of scaling. Conventional wisdom encourages us to use as few nodes as
possible when deploying these protocols: ``using more than $2f+1$ replicas for
$f$ failures is possible but illogical because it requires a larger quorum size
with no additional benefit''~\cite{zhang2018building}. While some protocols use
multiple leaders~\cite{mao2008mencius, moraru2013there, arun2017speeding}, the
number of leaders is fixed (typically $2f+1$ leaders to tolerate $f$ faults),
which only alleviates but does not solve the scalability problem.
%

BPaxos, on the other hand, employs a straightforward form of scaling to achieve
high throughput. A BPaxos deployment consists of a set of leaders, dependency
service nodes, proposers, acceptors, and replicas. We will see later that
dependency service nodes, acceptors, and replicas do not scale. This is why
conventional wisdom dictates using as few of these nodes as possible. However,
leaders and proposers operate independently from one another and are thus
``embarrassingly scalable''. Moreover, when we analyze the performance of
BPaxos, we find that these leaders and proposers are the throughput bottleneck.
By increasing the number of leaders and proposers, we increase the protocol's
throughput. Note that BPaxos does not horizontally scale forever. Scaling the
leaders and proposers shifts the bottleneck to other non-scalable components.
With scaling, BPaxos is able to achieve roughly double the peak throughput of
EPaxos, a state-of-the-art replication protocol.

This straightforward form of scaling has been largely overlooked because most
existing replication protocols tightly couple their components together. For
example, an EPaxos replica plays the role of a leader, a dependency service
node, an acceptor, \emph{and} a replica~\cite{arun2017speeding}. This tight
coupling has a number of advantages---e.g., messages sent between co-located
nodes do not have to traverse the network, redundant metadata can be coalesced,
fast paths can be taken to reduce latency, and so on. However, tight coupling
lumps together components that do not scale with components that do. This
prevents independently scaling bottleneck components. BPaxos' modularity is the
key enabling feature that allows us to perform independent scaling.

\paragraph{Simplicity}
MultiPaxos is notoriously difficult to understand, and sophisticated protocols
that improve it are significantly more complex. BPaxos' modular design, on the
other hand, makes the protocol much easier to understand compared to these
sophisticated protocols. Each module can be understood and proven correct in
isolation, allowing newcomers to understand the protocol piece by piece,
something that is difficult to do with existing protocols in which components
are tightly coupled.

Moreover, some of the modules implement well-known abstractions for which
well-established protocols already exist. In these cases, BPaxos can leverage
existing protocols instead of reinventing the wheel. For example, BPaxos
depends on a module that implements consensus. Rather than implementing a
consensus protocol from scratch and having to prove it correct, BPaxos uses
Paxos off the shelf and inherits its safety properties. Many other
protocols~\cite{moraru2013there, arun2017speeding, nawab2018dpaxos} instead
implement consensus in a way that is specialized to each protocol. These
specialized consensus protocols are difficult to understand and difficult to
prove correct. As an anecdote, we discovered a minor bug in EPaxos'
implementation of consensus, which we confirmed with the authors, a bug that
went undiscovered for six years.


\paragraph{Summary}
In summary, we present the following contributions:
\begin{itemize}
  \item
    We introduce BPaxos, a modular, multileader, generalized state machine
    replication protocol that is significantly easier to understand compared to
    similar protocols.
  \item
    We describe how modularity enables a straightforward form of protocol
    scaling. We apply the technique to BPaxos and achieve double the peak
    throughput of a state-of-the-art replication protocol.
\end{itemize}
}
{\section{Background}

\subsection{Paxos}
Assume we have a number of clients, each with a value that they would like to
propose. The \defword{consensus} problem is for all members to agree on a
single value among the proposed values. A consensus protocol is a protocol that
implements consensus. Clients propose commands by sending them to the protocol.
The protocol eventually chooses a single one of the proposed values and returns
it to the clients.

Paxos~\cite{lamport1998part, lamport2001paxos} is one of the oldest and most
well studied consensus protocols. We will see later that BPaxos uses Paxos to
implement consensus, so it is important to be familiar with \emph{what} Paxos
is. Fortunately though, BPaxos treats Paxos like a black box, so we do not have
to concern ourselves with \emph{how} Paxos works.

\subsection{MultiPaxos}
Whereas consensus protocols like Paxos agree on a \emph{single} value,
\defword{state machine replication} protocols like MultiPaxos agree on a
\emph{sequence} of values called a log.  A state machine replication protocol
involves some number of replicas of a state machine, with each state machine
beginning in the same initial state. Clients propose commands to the
replication protocol, and the protocol orders the commands into an agreed upon
log that grows over time. Replicas execute entries in the log in prefix order.
By beginning in the same initial state and executing the same commands in the
same order, all the replicas are guaranteed to remain in sync.

MultiPaxos~\cite{van2015paxos} is one of the earliest and most popular state
machine replication protocols. MultiPaxos uses one instance of Paxos for every
log entry, agreeing on the log entries one by one. For example, it runs one
instance of Paxos to agree on the command chosen in log entry 0, one instance
for log entry 1, and so on. Over time, more and more commands are chosen, and
the log of chosen commands grows and grows. MultiPaxos replicas execute
commands as they are chosen, taking care not to execute the commands out of
order.

For example, consider the example execution of a MultiPaxos replica depicted in
\figref{ExampleMultiPaxosExecution}. The replica implements a key-value store
with keys $a$ and $b$. First, the command $a \gets 0$ (i.e.\ set $a$ to $0$) is
chosen in log entry $0$ (\figref{ExampleMultiPaxosExecutionA}), and the replica
executes the command (\figref{ExampleMultiPaxosExecutionB}). Then, the command
$a \gets b$ is chosen in log entry $2$ (\figref{ExampleMultiPaxosExecutionC}).
The replica cannot yet execute the command, because it must first execute the
command in log entry $1$, which has not yet been chosen
(\figref{ExampleMultiPaxosExecutionD}).  Finally, $b \gets 0$ is chosen in log
entry $1$ (\figref{ExampleMultiPaxosExecutionE}), and the replica can execute
the commands in both log entries $1$ and $2$. Note that the replica executes
the log in prefix order, waiting to execute a command if previous commands have
not yet been chosen and executed.

{\newlength{\logentryinnersep}
\setlength{\logentryinnersep}{2pt}
\newlength{\logentrylinewidth}
\setlength{\logentrylinewidth}{1pt}
\newlength{\logentrywidth}
\setlength{\logentrywidth}{\widthof{\scriptsize$a \gets 1$}+2\logentryinnersep}
\newcommand{\logindexcolor}{flatred}
\newcommand{\cmdi}{$a \gets 0$}
\newcommand{\cmdii}{$b \gets 0$}
\newcommand{\cmdiii}{$a \gets b$}

\tikzstyle{logentry}=[draw,
                      font=\scriptsize,
                      inner sep=\logentryinnersep,
                      line width=\logentrylinewidth,
                      minimum height=\logentrywidth,
                      minimum width=\logentrywidth]
\tikzstyle{executed}=[fill=gray, opacity=0.2, draw opacity=1, text opacity=1]
\tikzstyle{logindex}=[\logindexcolor]

\newcommand{\rightof}[1]{-\logentrylinewidth of #1}

\newcommand{\multipaxoslog}[6]{%
  \node[logentry, label={[logindex]90:0}, #2] (0) {#1};
  \node[logentry, label={[logindex]90:1}, right=\rightof{0}, #4] (1) {#3};
  \node[logentry, label={[logindex]90:2}, right=\rightof{1}, #6] (2) {#5};
}

\begin{figure}[ht]
  \begin{subfigure}[t]{0.45\columnwidth}
    \centering
    \begin{tikzpicture}
      \multipaxoslog{\cmdi}{}%
                    {}{}%
                    {}{}
    \end{tikzpicture}
    \caption{%
      \cmdi{} is chosen in entry $\textcolor{\logindexcolor}{0}$.
    }
    \figlabel{ExampleMultiPaxosExecutionA}
  \end{subfigure}\hspace{0.1\columnwidth}%
  \begin{subfigure}[t]{0.45\columnwidth}
    \centering
    \begin{tikzpicture}
      \multipaxoslog{\cmdi}{executed}%
                    {}{}%
                    {}{}
    \end{tikzpicture}
    \caption{%
      \cmdi{} is executed.
    }
    \figlabel{ExampleMultiPaxosExecutionB}
  \end{subfigure}

  \vspace{2pt}\textcolor{flatgray}{\rule{\columnwidth}{0.4pt}}

  \begin{subfigure}[t]{0.45\columnwidth}
    \centering
    \begin{tikzpicture}
      \multipaxoslog{\cmdi}{executed}%
                    {}{}%
                    {\cmdiii}{}
    \end{tikzpicture}
    \caption{%
      \cmdiii{} is chosen in entry $\textcolor{\logindexcolor}{2}$.
    }
    \figlabel{ExampleMultiPaxosExecutionC}
  \end{subfigure}\hspace{0.1\columnwidth}%
  \begin{subfigure}[t]{0.45\columnwidth}
    \centering
    \begin{tikzpicture}
      \multipaxoslog{\cmdi}{executed}%
                    {}{}%
                    {\cmdiii}{}
    \end{tikzpicture}
    \caption{%
      Nothing is executed.
    }
    \figlabel{ExampleMultiPaxosExecutionD}
  \end{subfigure}

  \vspace{2pt}\textcolor{flatgray}{\rule{\columnwidth}{0.4pt}}

  \begin{subfigure}[t]{0.45\columnwidth}
    \centering
    \begin{tikzpicture}
      \multipaxoslog{\cmdi}{executed}%
                    {\cmdii}{}%
                    {\cmdiii}{}
    \end{tikzpicture}
    \caption{%
      \cmdii{} is chosen in entry $\textcolor{\logindexcolor}{1}$.
    }
    \figlabel{ExampleMultiPaxosExecutionE}
  \end{subfigure}\hspace{0.1\columnwidth}%
  \begin{subfigure}[t]{0.45\columnwidth}
    \centering
    \begin{tikzpicture}
      \multipaxoslog{\cmdi}{executed}%
                    {\cmdii}{executed}%
                    {\cmdiii}{executed}
    \end{tikzpicture}
    \caption{%
      \cmdii{}, \cmdiii{} are executed.
    }
    \figlabel{ExampleMultiPaxosExecutionF}
  \end{subfigure}

  \caption{%
    An example of a MultiPaxos replica executing commands over time, as they
    are chosen
  }
  \figlabel{ExampleMultiPaxosExecution}
\end{figure}
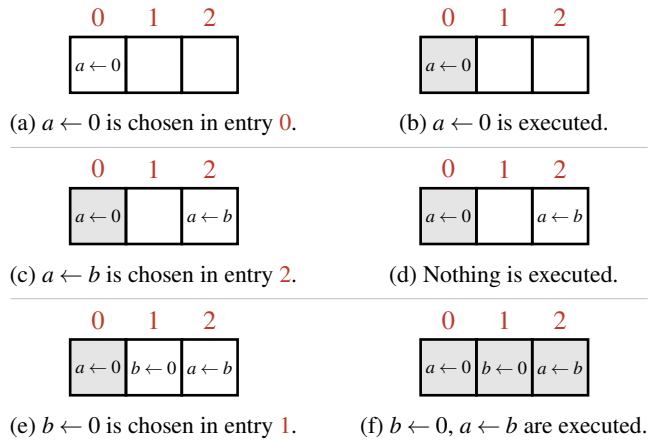
}

MultiPaxos is implemented with a set of nodes called proposers and a set of
nodes called acceptors. For this paper, we do not need to worry about the
details of how MultiPaxos works, but let us focus briefly on its communication
pattern. One of the proposers is designated a leader. Clients send all state
machine commands to this single leader. When the leader receives a command $x$,
it selects a log entry in which to place $x$ and then performs one round trip
of communication with the acceptors to get $x$ chosen in the log entry. Then,
it executes the command---once all commands in earlier log entries have been
chosen and executed---and returns to the client.  This communication pattern is
illustrated in \figref{MultiPaxosCommunication}.

{\input{figures/common.tex}

\tikzstyle{proc}=[draw, circle, thick]
\tikzstyle{proclabel}=[inner sep=0pt]
\tikzstyle{fakeproclabel}=[inner sep=0pt, white]
\tikzstyle{comm}=[-latex, thick]
\tikzstyle{commnum}=[fill=white, inner sep=0pt]
\newcommand{\acceptorcolor}{flatorange}

\begin{figure}[ht]
  \centering
  \begin{tikzpicture}[scale=0.75]
    \node[proc] (c) at (0, 2) {$c$};
    \node[proc, fill=\proposercolor!25] (p0) at (2, 2) {$p_0$};
    \node[proc, fill=\proposercolor!25] (p1) at (4, 2) {$p_1$};
    \node[proc, fill=\acceptorcolor!25] (a0) at (0, 0) {$a_0$};
    \node[proc, fill=\acceptorcolor!25] (a1) at (2, 0) {$a_1$};
    \node[proc, fill=\acceptorcolor!25] (a2) at (4, 0) {$a_2$};

    \crown{(p0.north)++(0, -0.15)}{1}{0.5}

    \node[proclabel, anchor=east] at (-0.5, 2) {Client};
    \node[fakeproclabel, anchor=west] (ps) at (5, 2) {Proposers};
    \halffill{ps}{\proposercolor!25}
    \node[proclabel, anchor=west] (ps) at (5, 2) {Proposers};
    \node[fakeproclabel, anchor=west] (as) at (5, 0) {Acceptors};
    \halffill{as}{\acceptorcolor!25}
    \node[proclabel, anchor=west] (as) at (5, 0) {Acceptors};

    \draw[comm, bend left=10] (c) to node[commnum, midway]{1} (p0);
    \draw[comm, bend left=10] (p0) to node[commnum, midway]{2} (a0);
    \draw[comm, bend left=10] (p0) to node[commnum, midway]{2} (a1);
    \draw[comm, bend left=10] (p0) to node[commnum, midway]{2} (a2);
    \draw[comm, bend left=10] (a0) to node[commnum, midway]{3} (p0);
    \draw[comm, bend left=10] (a1) to node[commnum, midway]{3} (p0);
    \draw[comm, bend left=10] (a2) to node[commnum, midway]{3} (p0);
    \draw[comm, bend left=10] (p0) to node[commnum, midway]{4} (c);
  \end{tikzpicture}
  \caption{%
    MultiPaxos communication pattern. The leader is adorned with a crown.
  }
  \figlabel{MultiPaxosCommunication}
\end{figure}
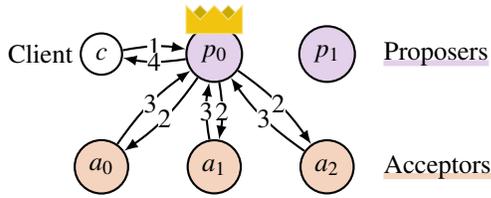}

\subsection{Multileader and Generalized Consensus}
MultiPaxos has a number of inefficiencies. Here, we focus on two well-known
ones. First, MultiPaxos' throughput is bottlenecked by the leader. As shown in
\figref{MultiPaxosCommunication}, \emph{every} command goes through the leader.
Thus, MultiPaxos can run only as fast as the leader can. Protocols like
Mencius~\cite{mao2008mencius}, EPaxos~\cite{moraru2013there}, and
Caesar~\cite{arun2017speeding} bypass the single leader bottleneck by having
multiple leaders that can process requests in parallel. We call these protocols
\defword{multileader} protocols.

Second, MultiPaxos requires that replicas execute \emph{all} commands in the
same order. That is, MultiPaxos establishes a \emph{total order} of commands.
This is overkill. If two commands commute, they can be executed by replicas in
either order. For example, key-value store replicas executing the log in
\figref{ExampleMultiPaxosExecution} could execute commands $a \gets 0$ and $b
\gets 0$ in either order since the two commands commute. More formally, we say
two commands \defword{conflict} if executing them in opposite orders yields
either different outputs or a different final state. State machine replication
protocols that only require \emph{conflicting} commands to be executed in the
same order are said to implement generalized
consensus~\cite{lamport2005generalized}. Colloquially, we say such a protocol
is \defword{generalized}. Generalized protocols establish a \emph{partial
order} of commands (as opposed to a total order) in which only conflicting
commands have to be ordered.

As a MultiPaxos leader receives commands from clients, it places them in
increasing log entries. The first command is placed in log entry 0, the second
in log entry 1, and so on. In this way, the leader acts as a sequencer,
sequencing commands into a single total order. Multileader protocols however, by
virtue of having multiple leaders, do not have a single designated node that
processes every command. This makes it challenging to establish a single total
order. As a result, most multileader protocols are also generalized. With
multiple concurrently executing leaders, it is easier to establish a partial
order than it is to establish a total order. Moreover, generalization allows
leaders processing non-conflicting commands to operate completely independently
from one another. While it is possible for a multileader protocol to establish
a total order (e.g., Mencius~\cite{mao2008mencius}), such protocols run only as
fast as the slowest replica (which lowers throughput), and involve all-to-all
communication among the leaders (which also lowers throughput).
}
{\section{Bipartisan Paxos}
BPaxos is a modular state machine replication protocol that is both multileader
and generalized. Throughout the paper, we make the standard assumptions that
the network is asynchronous, that state machines are deterministic, and that
machines can fail by crashing but cannot act maliciously. We also assume that
at most $f$ machines can fail for some integer-valued parameter $f$. Throughout
the paper, we omit low-level protocol details involving the re-sending of
dropped messages.

\subsection{BPaxos Command Execution}
MultiPaxos is \emph{not} generalized. It totally orders all commands by
sequencing them into a \emph{log}. BPaxos is generalized, so it ditches the log
and instead partially orders commands into a \emph{directed graph}, like the
ones shown in \figref{ExampleBPaxosExecution}.

BPaxos graphs are completely analogous to MultiPaxos logs. Every MultiPaxos log
entry corresponds to a \defword{vertex} in a BPaxos graph. Every MultiPaxos log
entry holds a command; so does every vertex. Every log entry is uniquely
identified by its index (e.g., \textcolor{flatred}{$0$}); every vertex is
uniquely identified by a \defword{vertex id} (e.g.,
\textcolor{flatred}{$v_0$}). The one difference between graphs and logs are the
edges. Every BPaxos vertex $v$ has edges to some set of other vertices. These
edges are called the \defword{dependencies} of $v$. Note that we view a
vertex's dependencies as belonging to the vertex, so when we refer to a vertex,
we are also referring to its dependencies. The similarities between MultiPaxos
logs and BPaxos graphs are summarized in \tabref{MultiPaxosVsBPaxos}.

\begin{table}[ht]
  \centering
  \caption{A comparison of MultiPaxos log entries and BPaxos vertices.}
  \tablabel{MultiPaxosVsBPaxos}
  \begin{tabular}{r|l}
    \textbf{BPaxos} & \textbf{MultiPaxos} \\\hline
    graph           & log \\
    vertex          & log entry \\
    vertex id       & index \\
    command         & command \\
    dependencies    & - \\
  \end{tabular}
\end{table}

MultiPaxos grows its \emph{log} over time by repeatedly reaching consensus on
one \emph{log entry} at a time. BPaxos grows its \emph{graph} over time by
repeatedly reaching consensus on one \emph{vertex} (and its dependencies) at a
time. MultiPaxos replicas execute logs in prefix order, making sure not to
execute a command until after executing \emph{all previous commands}. BPaxos
replicas execute graphs in prefix order (i.e. reverse topological order),
making sure not to execute a command until after executing \emph{its
dependencies}.

An example of how BPaxos graphs grow over time and how a BPaxos replica
executes these graphs in shown in \figref{ExampleBPaxosExecution}. As you read
through the figure, note the similarities with
\figref{ExampleMultiPaxosExecution}.
First, the command $a \gets 0$ is chosen in vertex $v_0$ with no dependencies
(\figref{ExampleBPaxosExecutionA}).
Because the vertex has no dependencies, the replica executes $a \gets 0$
immediately (\figref{ExampleBPaxosExecutionB}).
Next, the command $a \gets b$ is chosen in vertex $v_2$ with dependencies on
vertices $v_0$ and $v_1$ (\figref{ExampleBPaxosExecutionC}).
$v_2$ depends on $v_1$, but a command has not yet been chosen in $v_1$, so the
replica does \emph{not} yet execute $a \gets b$
(\figref{ExampleBPaxosExecutionD}).
Finally, the command $b \gets 0$ is chosen in vertex $v_1$ with no
dependencies (\figref{ExampleBPaxosExecutionE}).
Because $b \gets 0$ has no dependencies, the replica executes it immediately.
Moreover, all of $v_2$'s dependencies have been executed, so the replica now
executes $a \gets b$ (\figref{ExampleBPaxosExecutionF}).

{\newlength{\vertexinnersep}
\setlength{\vertexinnersep}{2pt}
\newlength{\vertexlinewidth}
\setlength{\vertexlinewidth}{1pt}
\newlength{\vertexwidth}
\setlength{\vertexwidth}{\widthof{\scriptsize $a \gets 1$}+2\vertexinnersep}
\newcommand{\graphindexcolor}{flatred}
\newcommand{\cmdi}{$a \gets 0$}
\newcommand{\cmdii}{$b \gets 0$}
\newcommand{\cmdiii}{$a \gets b$}
\newcommand{\xscale}{0.75}
\newcommand{\yscale}{0.75}

\tikzstyle{vertex}=[draw,
                    font=\scriptsize,
                    inner sep=\vertexinnersep,
                    line width=\vertexlinewidth,
                    minimum height=\vertexwidth,
                    minimum width=\vertexwidth]

\tikzstyle{executed}=[fill=gray, opacity=0.2, draw opacity=1, text opacity=1]

\tikzstyle{dep}=[-latex, thick]

\tikzstyle{graphindex}=[\graphindexcolor]

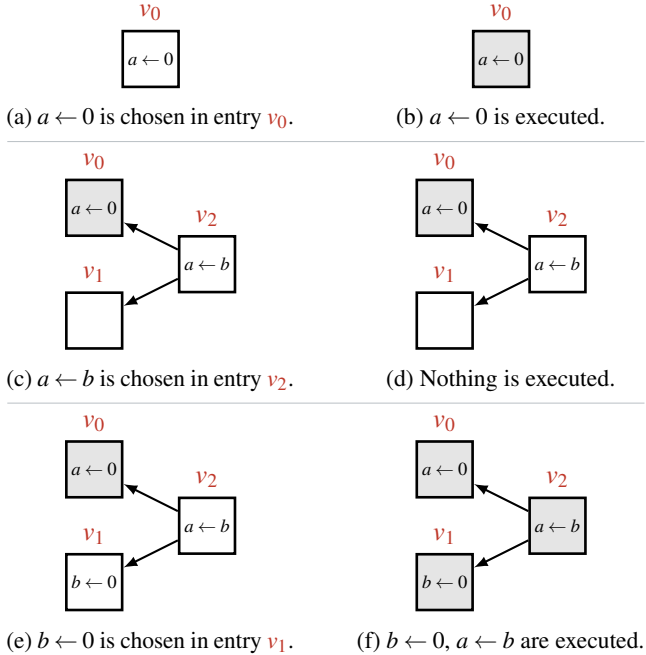
\begin{figure}
  \begin{subfigure}[t]{0.45\columnwidth}
    \centering
    \begin{tikzpicture}[xscale=\xscale, yscale=\yscale]
      \node[vertex, label={[graphindex]90:$v_0$}] (0) at (0, 2) {\cmdi{}};
    \end{tikzpicture}
    \caption{%
      \cmdi{} is chosen in entry $\textcolor{\graphindexcolor}{v_0}$.
    }
    \figlabel{ExampleBPaxosExecutionA}
  \end{subfigure}\hspace{0.1\columnwidth}%
  \begin{subfigure}[t]{0.45\columnwidth}
    \centering
    \begin{tikzpicture}[xscale=\xscale, yscale=\yscale]
      \node[vertex, executed, label={[graphindex]90:$v_0$}] (0) at (0, 2)
        {\cmdi{}};
    \end{tikzpicture}
    \caption{%
      \cmdi{} is executed.
    }
    \figlabel{ExampleBPaxosExecutionB}
  \end{subfigure}

  \vspace{2pt}\textcolor{flatgray}{\rule{\columnwidth}{0.4pt}}

  \begin{subfigure}[t]{0.45\columnwidth}
    \centering
    \begin{tikzpicture}[xscale=\xscale, yscale=\yscale]
      \node[vertex, executed, label={[graphindex]90:$v_0$}] (0) at (0, 2)
        {\cmdi{}};
      \node[vertex, label={[graphindex]90:$v_2$}] (2) at (2, 1) {\cmdiii{}};
      \node[vertex, label={[graphindex]90:$v_1$}] (1) at (0, 0) {};
      \draw[dep] (2) to (0);
      \draw[dep] (2) to (1);
    \end{tikzpicture}
    \caption{%
      \cmdiii{} is chosen in entry $\textcolor{\graphindexcolor}{v_2}$.
    }
    \figlabel{ExampleBPaxosExecutionC}
  \end{subfigure}\hspace{0.1\columnwidth}%
  \begin{subfigure}[t]{0.45\columnwidth}
    \centering
    \begin{tikzpicture}[xscale=\xscale, yscale=\yscale]
      \node[vertex, executed, label={[graphindex]90:$v_0$}] (0) at (0, 2)
        {\cmdi{}};
      \node[vertex, label={[graphindex]90:$v_2$}] (2) at (2, 1) {\cmdiii{}};
      \node[vertex, label={[graphindex]90:$v_1$}] (1) at (0, 0) {};
      \draw[dep] (2) to (0);
      \draw[dep] (2) to (1);
    \end{tikzpicture}
    \caption{%
      Nothing is executed.
    }
    \figlabel{ExampleBPaxosExecutionD}
  \end{subfigure}

  \vspace{2pt}\textcolor{flatgray}{\rule{\columnwidth}{0.4pt}}

  \begin{subfigure}[t]{0.45\columnwidth}
    \centering
    \begin{tikzpicture}[xscale=\xscale, yscale=\yscale]
      \node[vertex, executed, label={[graphindex]90:$v_0$}] (0) at (0, 2)
        {\cmdi{}};
      \node[vertex, label={[graphindex]90:$v_2$}] (2) at (2, 1) {\cmdiii{}};
      \node[vertex, label={[graphindex]90:$v_1$}] (1) at (0, 0) {\cmdii{}};
      \draw[dep] (2) to (0);
      \draw[dep] (2) to (1);
    \end{tikzpicture}
    \caption{%
      \cmdii{} is chosen in entry $\textcolor{\graphindexcolor}{v_1}$.
    }
    \figlabel{ExampleBPaxosExecutionE}
  \end{subfigure}\hspace{0.1\columnwidth}%
  \begin{subfigure}[t]{0.45\columnwidth}
    \centering
    \begin{tikzpicture}[xscale=\xscale, yscale=\yscale]
      \node[vertex, executed, label={[graphindex]90:$v_0$}] (0) at (0, 2)
        {\cmdi{}};
      \node[vertex, executed, label={[graphindex]90:$v_2$}] (2) at (2, 1)
        {\cmdiii{}};
      \node[vertex, executed, label={[graphindex]90:$v_1$}] (1) at (0, 0)
        {\cmdii{}};
      \draw[dep] (2) to (0);
      \draw[dep] (2) to (1);
    \end{tikzpicture}
    \caption{%
      \cmdii{}, \cmdiii{} are executed.
    }
    \figlabel{ExampleBPaxosExecutionF}
  \end{subfigure}

  \caption{%
    An example of a BPaxos replica executing commands over time, as they are
    chosen.
  }
  \figlabel{ExampleBPaxosExecution}
\end{figure}
}

Before we discuss the mechanisms that BPaxos uses to construct these graphs,
note the following three graph properties.

\paragraph{Vertices are chosen once and for all.}
BPaxos reaches consensus on every vertex, so once a vertex has been chosen, it
will never change. Its command will not change, it will not lose dependencies,
and it will not get new dependencies.

\paragraph{Cycles can happen, but are not a problem.}
We will see in a moment that BPaxos graphs can sometimes be cyclic. These cycles
are a nuisance, but easily handled. Instead of executing graphs in reverse
topological order one \emph{command} at a time, replicas instead execute graphs
in reverse topological order one \emph{strongly connected component} at a time.
The commands within a strongly connected component are executed in an arbitrary
yet deterministic order (e.g., in vertex id order). This is illustrated in
\figref{ExampleBPaxosCycleExecution}.

{\newlength{\cyclevertexinnersep}
\setlength{\cyclevertexinnersep}{4pt}
\newlength{\cyclevertexlinewidth}
\setlength{\cyclevertexlinewidth}{1pt}
\newlength{\cyclevertexwidth}
\setlength{\cyclevertexwidth}{\widthof{$x$}+2\cyclevertexinnersep}
\newcommand{\graphindexcolor}{flatred}
\newcommand{\cmdi}{$a \gets 0$}
\newcommand{\cmdii}{$b \gets 0$}
\newcommand{\cmdiii}{$a \gets b$}
\newcommand{\xscale}{0.5}
\newcommand{\yscale}{0.5}

\tikzstyle{vertex}=[draw,
                    inner sep=\cyclevertexinnersep,
                    line width=\cyclevertexlinewidth,
                    minimum height=\cyclevertexwidth,
                    minimum width=\cyclevertexwidth]

\tikzstyle{executed}=[fill=gray, opacity=0.2, draw opacity=1, text opacity=1]

\tikzstyle{dep}=[-latex, thick]

\tikzstyle{graphindex}=[\graphindexcolor]

\begin{figure}
  \centering

  \begin{subfigure}[t]{0.3\columnwidth}
    \centering
    \begin{tikzpicture}[xscale=\xscale, yscale=\yscale]
      \node[vertex, executed, label={[graphindex]90:$v_x$}] (x) at (0, 1) {$x$};
      \node[vertex, label={[graphindex]90:$v_y$}] (y) at (2, 2) {$y$};
      \node[vertex, label={[graphindex]-90:$v_z$}] (z) at (2, 0) {};
      \draw[dep] (y) to (x);
      \draw[dep, bend left] (y) to (z);
    \end{tikzpicture}
    \caption{}
    \figlabel{ExampleBPaxosCycleExecutionA}
  \end{subfigure}\hspace{0.04\columnwidth}%
  \begin{subfigure}[t]{0.3\columnwidth}
    \centering
    \begin{tikzpicture}[xscale=\xscale, yscale=\yscale]
      \node[vertex, executed, label={[graphindex]90:$v_x$}] (x) at (0, 1) {$x$};
      \node[vertex, label={[graphindex]90:$v_y$}] (y) at (2, 2) {$y$};
      \node[vertex, label={[graphindex]-90:$v_z$}] (z) at (2, 0) {$z$};
      \draw[dep] (y) to (x);
      \draw[dep, bend left] (y) to (z);
      \draw[dep, bend left] (z) to (y);
    \end{tikzpicture}
    \caption{}
    \figlabel{ExampleBPaxosCycleExecutionB}
  \end{subfigure}\hspace{0.04\columnwidth}%
  \begin{subfigure}[t]{0.3\columnwidth}
    \centering
    \begin{tikzpicture}[xscale=\xscale, yscale=\yscale]
      \node[vertex, executed, label={[graphindex]90:$v_x$}] (x) at (0, 1) {$x$};
      \node[vertex, executed, label={[graphindex]90:$v_y$}] (y) at (2, 2) {$y$};
      \node[vertex, executed, label={[graphindex]-90:$v_z$}] (z) at (2, 0) {$z$};
      \draw[dep] (y) to (x);
      \draw[dep, bend left] (y) to (z);
      \draw[dep, bend left] (z) to (y);
    \end{tikzpicture}
    \caption{}
    \figlabel{ExampleBPaxosCycleExecutionB}
  \end{subfigure}

  \caption{
    An example of a BPaxos replica executing a cyclic graph.
    (a) $y$ cannot be exeucted until $v_z$ is chosen.
    (b) $v_z$ is chosen. $v_y$ and $v_z$ form a strongly connected component.
    (c) $y$ and $z$ are executed in an arbitrary yet deterministic order; $y$
        then $z$ or $z$ then $y$.%
  }
  \figlabel{ExampleBPaxosCycleExecution}
\end{figure}
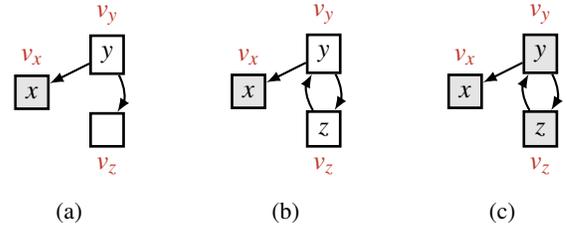
}

\paragraph{Conflicting commands depend on each other.}
Because BPaxos is generalized, only conflicting commands have to be ordered
with respect to each other. BPaxos ensures this by maintaining the following
invariant:
\begin{invariant}[\defword{dependency invariant}]\invlabel{KeyInvariant}
  If two conflicting commands $x$ and $y$ are chosen in vertices $v_x$ and
  $v_y$, then either $v_x$ depends on $v_y$ or $v_y$ depends on $v_x$ or both.
  That is, there is at least one edge between vertices $v_x$ and $v_y$.
\end{invariant}
If two commands have an edge between them, every replica executes them in the
same order.  The dependency invariant ensures that every conflicting pair of
commands has an edge between them, ensuring that all conflicting commands are
executed in the same order. Non-conflicting commands do not need an edge
between them and can be executed in any order.

\subsection{Protocol Overview}
BPaxos is composed of five modules: a dependency service, a consensus service,
a set of leaders, a set of proposers, and a set of replicas. Here, we give an
overview on how these modules interact by walking through the example execution
shown in \figref{BPaxosOverview}. In the next couple of sections, we discuss
each module in more detail.

1. A client $c$ sends a state machine command $x$ to leader $l_0$. Note that
all of the leaders process commands in parallel and that clients can send
commands to any of them.

2. Upon receiving command $x$, $l_0$ generates a globally unique vertex id
$v_x$ for $x$. It then sends the message $\msg{v_x, x}$ to the dependency
service.

3. Upon receiving message $\msg{v_x, x}$, the dependency service computes a set
of dependencies $\deps{}_x$ for vertex $v_x$. Later, we will see exactly how the
dependency service computes dependencies. For now, we overlook the details. The
dependency service then sends back the message $\msg{v_x, x, \deps{}_x}$ to
$l_0$.

4. $l_0$ forwards $\msg{v_x, x, \deps{}_x}$ to proposer $p_0$.

5. $p_0$ sends the message $\msg{v_x, x, \deps{}_x}$ to the consensus service,
proposing that the value $(x, \deps{}_x)$ be chosen in vertex $v_x$.

6. The consensus service implements one instance of consensus for every vertex.
Upon receiving $\msg{v_x, x, \deps{}_x}$, it chooses the value $(x, \deps{}_x)$
for vertex $v_x$ and notifies $p_0$ with the message $\msg{v_x, x, \deps{}_x}$.
Note that in this example, the consensus service chose the value proposed by
$p_0$. In general, the consensus service may choose some other value if other
proposers are concurrently proposing different values for vertex $v_x$.
However, we will see later that this can only happen during recovery and is
therefore not typical.

7. After $p_0$ learns that command $x$ with dependencies $\deps{}_x$ has been
chosen in vertex $v_x$, it notifies the replicas by broadcasting the message
$\msg{v_x, x, \deps{}_x}$.

8. Every replica manages a graph of chosen commands, as described in the
previous subsection. Upon receiving $\msg{v_x, x, \deps{}_x}$, a replica adds
the vertex $v_x$ to its graph with command $x$ and dependencies $\deps{}_x$.
As described earlier, the replicas execute their graphs in reverse topological
order. Once they have executed command $x$, yielding output $o$, one of the
replicas sends back the response to the client $c$. Given $r$ replicas, replica
$i$ sends back the response where $i = \text{hash}(v_x) \% r$ for some hash
function.

Pseudocode for BPaxos is given in \figref{BPaxosPseudocode}, and a TLA+
specification of BPaxos is given in \appref{TlaSpec}. We now detail each BPaxos
module. In the next section, we discuss why the dependency service, consensus
service, and replicas do not scale and why the leaders and proposers do.

{\input{figures/common.tex}

\tikzstyle{proc}=[draw, circle, thick, inner sep=2pt]
\tikzstyle{leader}=[proc, fill=\leadercolor!25]
\tikzstyle{proposer}=[proc, fill=\proposercolor!25]
\tikzstyle{replica}=[proc, fill=\replicacolor!25]
\tikzstyle{client}=[proc, fill=\clientcolor!25]
\tikzstyle{proclabel}=[inner sep=0pt]
\tikzstyle{fakeproclabel}=[inner sep=0pt, white]
\tikzstyle{comm}=[-latex, thick]
\tikzstyle{commnum}=[fill=white, inner sep=0pt]
\tikzstyle{service}=[draw, rounded corners, align=center, thick]
\tikzstyle{depservice}=[service, draw=\depservicecolor!50]
\tikzstyle{consensus}=[service, draw=\consensuscolor]
\tikzstyle{module}=[draw, thick, flatgray, rounded corners]

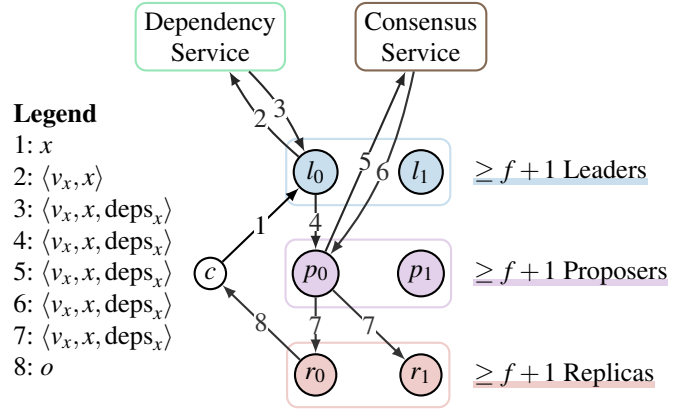
\begin{figure}[t]
  \centering

  \begin{tikzpicture}[yscale=0.9, xscale=0.70]
    \node[client] (c) at (0, 1.5) {$c$};
    \node[leader] (l0) at (2, 3) {$l_0$};
    \node[leader] (l1) at (4, 3) {$l_1$};
    \node[proposer] (p0) at (2, 1.5) {$p_0$};
    \node[proposer] (p1) at (4, 1.5) {$p_1$};
    \node[replica] (r0) at (2, 0) {$r_0$};
    \node[replica] (r1) at (4, 0) {$r_1$};
    \node[depservice] (depservice) at (0, 5) {Dependency\\Service};
    \node[consensus] (consensus) at (4, 5) {Consensus\\Service};

    \draw[module, \leadercolor!25]
      ($(l0.south west) + (-0.25, -0.25)$) rectangle
      ($(l1.north east) + (0.25, 0.25)$);
      \draw[module, \proposercolor!25]
      ($(p0.south west) + (-0.25, -0.25)$) rectangle
      ($(p1.north east) + (0.25, 0.25)$);
    \draw[module, \replicacolor!25]
      ($(r0.south west) + (-0.25, -0.25)$) rectangle
      ($(r1.north east) + (0.25, 0.25)$);

    \draw[black!100, comm] (c) to node[commnum]{$1$} (l0);
    \draw[black!85, comm, bend left=10] (l0) to node[commnum]{$2$} (depservice);
    \draw[black!85, comm, bend left=10] (depservice) to node[commnum]{$3$} (l0);
    \draw[black!85, comm] (l0) to node[commnum]{$4$} (p0);
    \draw[black!85, comm] (p0) to node[commnum]{$5$} (consensus);
    \draw[black!80, comm, bend left=15] (consensus) to node[commnum]{$6$} (p0);
    \draw[black!80, comm] (p0) to node[commnum]{$7$} (r0);
    \draw[black!80, comm] (p0) to node[commnum]{$7$} (r1);
    \draw[black!80, comm] (r0) to node[commnum]{$8$} (c);

    \node[fakeproclabel, anchor=west] (leaders) at (5, 3) {$\geq f+1$ Leaders};
    \halffill{leaders}{\leadercolor!25}
    \node[proclabel, anchor=west] at (leaders.west) {$\geq f+1$ Leaders};

    \node[fakeproclabel, anchor=west] (proposers) at (5, 1.5) {$\geq f+1$ Proposers};
    \halffill{proposers}{\proposercolor!25}
    \node[proclabel, anchor=west] at (proposers.west) {$\geq f+1$ Proposers};

    \node[fakeproclabel, anchor=west] (replicas) at (5, 0) {$\geq f+1$ Replicas};
    \halffill{replicas}{\replicacolor!25}
    \node[proclabel, anchor=west] at (replicas.west) {$\geq f+1$ Replicas};

    \node[align=left, anchor=east] at ($(c.west)+(-0.2, 0.5)$) {%
      \textbf{Legend} \\
      $1$: $x$ \\
      $2$: $\msg{v_x, x}$ \\
      $3$: $\msg{v_x, x, \deps{}_x}$ \\
      $4$: $\msg{v_x, x, \deps{}_x}$ \\
      $5$: $\msg{v_x, x, \deps{}_x}$ \\
      $6$: $\msg{v_x, x, \deps{}_x}$ \\
      $7$: $\msg{v_x, x, \deps{}_x}$ \\
      $8$: $o$
    };
  \end{tikzpicture}

  \caption{%
    An overview of BPaxos execution. Note that we show the execution of only a
    single command for simplicity, so only one leader and one proposer are
    active. In a real BPaxos deployment, there are multiple clients and every
    leader and every proposer is active.
  }\figlabel{BPaxosOverview}
\end{figure}}
{
\newcommand{\LineComment}[1]{\textcolor{gray}{// #1}}

\tikzstyle{comm}=[-latex, thick]
\tikzstyle{commnum}=[fill=white, inner sep=1pt]
\tikzstyle{proc}=[thick, text width=0.45\textwidth, anchor=north]

\begin{figure*}[t]
  \centering
  \begin{tikzpicture}
    \node[draw,
          thick,
          minimum width=6.8in,
          minimum height=0.25in]
          (clients) at (4.5, 1.25) {\large \textbf{Clients}};

    \node[proc,
          draw=\leadercolor,
          label={90:\textbf{\large Leader}}] (leader) at (0, 0) {%
      \newcommand{\leaderindex}{\textsf{leader index}}
      \newcommand{\nextid}{\textsf{next id}}
      \begin{algorithmic}[1]
        \GlobalState $\leaderindex{}$ \LineComment{e.g., $l_0$ has index $0$}
        \GlobalState $\nextid{} \gets 0$
        \Upon{receiving command $x$ from client}
          \State $v_x \gets (\leaderindex{}, \nextid{})$
          \State $\nextid{} \gets \nextid + 1$
          \State send $\msg{v_x, x}$ to dependency service nodes
        \EndUpon
        \Upon{receiving dependencies from $f + 1$ dependency service nodes for vertex $v_x$}
          \State let $\deps{}_1, \ldots, \deps{}_{f+1}$ be the dependencies
          \State $\deps{}_x \gets \bigcup_i \deps{}_i$
          \State send $\msg{v_x, x, \deps{}_x}$ to a proposer
        \EndUpon
      \end{algorithmic}
    };

    \node[proc,
          draw=\depservicecolor,
          label={90:\textbf{\large Dependency Service Node}},
          anchor=north] (depservice) at (9, 0) {%
      \newcommand{\commands}{\textsf{cmds}}
      \begin{algorithmic}[1]
        \GlobalState $\commands$ \LineComment{set of messages $\msg{v_x, x}$}
        \Upon{receiving $\msg{v_x, x}$ from leader $l$}
          \State $\deps{} = \setst{v_y}{\msg{v_y, y} \in \commands \land
                                        \text{$x$, $y$ conflict}}$
          \State $\commands{} \gets \commands \cup \set{\msg{v_x, x}}$
          \State send $\msg{v_x, x, \deps{}}$ to $l$
        \EndUpon
      \end{algorithmic}
    };

    \node[proc,
          draw=\proposercolor,
          label={90:\textbf{\large Proposer}},
          below=of leader] (proposer) {%
      \begin{algorithmic}[1]
        \Upon{receiving $\msg{v_x, x, \deps{}_x}$ from leader}
          \State send $\msg{v_x, x, \deps{}_x}$ to consensus service
        \EndUpon

        \Upon{receiving $\msg{v_x, x, \deps{}_x}$ from consensus service}
          \State send $\msg{v_x, x, \deps{}_x}$ to replicas
        \EndUpon
      \end{algorithmic}
    };

    \node[proc,
          draw=\consensuscolor,
          label={90:\textbf{\large Consensus Service}},
          below=of depservice] (consensus) {%
      \begin{algorithmic}[1]
        \Upon{receiving $\msg{v_x, x, \deps{}_x}$ from proposer $p$}
          \State reach consensus on $(x', \deps{}_x')$ for vertex $v_x$
          \State send $\msg{v_x, x', \deps{}_x'}$ to $p$
        \EndUpon
      \end{algorithmic}
    };

    \node[proc,
          draw=\replicacolor,
          label={90:\textbf{\large Replica}},
          below=of consensus] (replica) {%
      \newcommand{\graph}{\textsf{graph}}
      \newcommand{\replicaIndex}{\textsf{replica index}}
      \newcommand{\numReplicas}{\textsf{num replicas}}
      \begin{algorithmic}[1]
        \GlobalState $\graph$ \LineComment{BPaxos graph of chosen vertices}
        \GlobalState $\numReplicas{}$ \LineComment{the number of replicas}
        \Upon{receiving $\msg{v_x, x, \deps{}_x}$ from proposer}
          \State add $\msg{v_x, x, \deps{}_x}$ to $\graph$
          \State execute every eligible vertex $v_y$
          \If{$\text{hash}(v_y)~\%~\numReplicas = \replicaIndex$}
            \State send result of executing $v_y$ back to client
          \EndIf
        \EndUpon
      \end{algorithmic}
    };

    \draw[comm]
      ($(clients.south west)!0.72!(clients.south)$)
      to node[commnum] {$1$}
      ($(leader.north)!0.5!(leader.north east)$);
    \draw[comm, black!85]
      ($(leader.north east)!0.5!(leader.east)$)
      to node[commnum] {$2$}
      (depservice);
    \draw[comm, black!85]
      (depservice.south west) to node[commnum] {$3$} (leader.east);
    \draw[comm, black!85]
      ($(leader.south)!0.5!(leader.south east)$)
      to node[commnum] {$4$}
      ($(proposer.north)!0.5!(proposer.north east)$);
    \draw[comm, black!80] (proposer.north east) to node[commnum] {$5$} (consensus.west);
    \draw[comm, black!80] (consensus) to node[commnum] {$6$} (proposer.north east);
    \draw[comm, black!80] (proposer) to node[commnum] {$7$} (replica);
    \draw[comm, black!80]
      (replica.east) to ++(0.5, 0)
                     to node[commnum]{$8$} ++(0, 8.75)
                     to (clients.east);
  \end{tikzpicture}

  \caption{BPaxos pseudocode}
  \figlabel{BPaxosPseudocode}
\end{figure*}
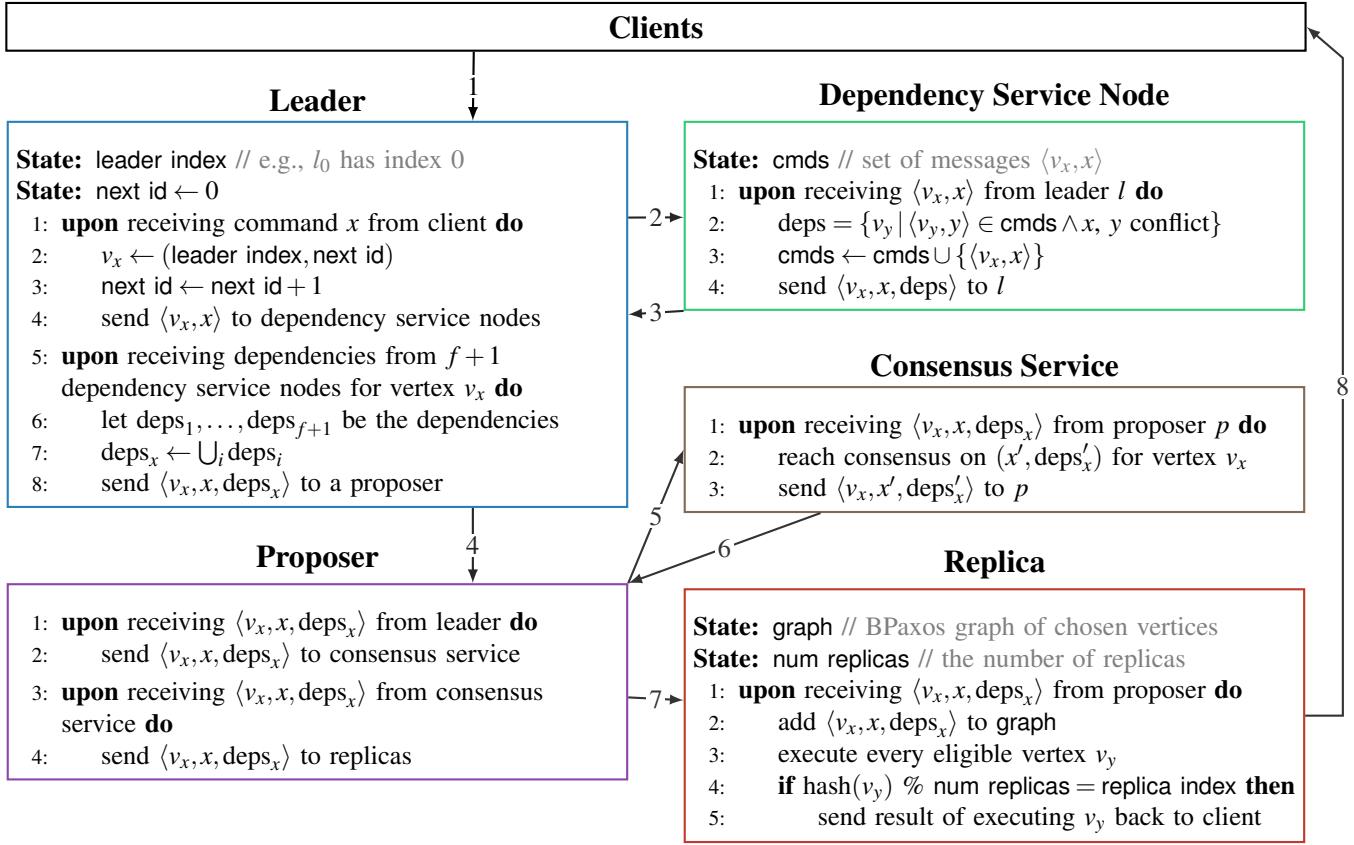}

\subsection{Dependency Service}
When the dependency service receives a message of the form $\msg{v_x, x}$, it
replies with a set of dependencies $\deps{}_x$ for $v_x$ using the message
$\msg{v_x, x, \deps{}_x}$.

Concretely, we implement the dependency service with $2f+1$ dependency service
nodes. Every dependency service node maintains a single piece of state,
\textsf{commands}. \textsf{commands} is the set of all the messages that the
dependency service node has received to date. When a dependency service node
receives message $\msg{v_x, x}$ from a leader, it computes the dependencies of
$v_x$ as the set of all vertices $v_y$ in $\textsf{commands}$ that contain a
command that conflicts with $x$:
\[
  \deps{} = \setst{v_y}{\msg{v_y, y} \in \textsf{commands}
                        ~\text{and $x$ and $y$ conflict}}.
\]
It then adds $\msg{v_x, x}$ to \textsf{commands} and sends $\msg{v_x, x,
\deps{}}$ back to the leader. When a leader sends a message $\msg{v_x, x}$ to
the dependency service, it sends it to every dependency service node. Upon
receiving $f + 1$ responses, $\set{\msg{v_x, x, \deps{}_1}, \ldots, \msg{v_x,
x, \deps{}_{f+1}}}$, the leader computes the final dependencies as
$\bigcup_{i=1}^{f+1} \deps{}_i$, the union of the computed dependencies.

The dependency service maintains the following invariant.

\begin{invariant}[\defword{dependency service invariant}]%
  \invlabel{DepServiceInvariant}
  If the dependency service produces responses $\msg{v_x, x, \deps{}_x}$ and
  $\msg{v_y, y, \deps{}_y}$ for conflicting commands $x$ and $y$, then $v_x \in
  \deps{}_y$ or $v_y \in \deps{}_x$ or both.
\end{invariant}

That is, the dependency service computes dependencies such that conflicting
commands depend on each other. Note that the dependency service invariant
(\invref{DepServiceInvariant}) is very similar to the dependency invariant
(\invref{KeyInvariant}). This is not an accident. Only dependencies computed by
the dependency service can be chosen, so the dependency service invariant
suffices to guarantee that the dependency invariant is maintained.

\begin{theorem}\thmlabel{DepServiceInvariant}
  The dependency service maintains \invref{DepServiceInvariant}.
\end{theorem}
\begin{proof}
  Assume the dependency service produces responses $\msg{v_x, x, \deps{}_x}$ and
  $\msg{v_y, y, \deps{}_y}$ for conflicting commands $x$ and $y$. We want to
  show that $v_x \in \deps{}_y$ or $v_y \in \deps{}_x$ or both. $\deps{}_x$ is
  the union of dependencies computed by some set $Q_x$ of $f + 1$ dependency
  service nodes. Similarly, $\deps{}_y$ is the union of dependencies computed
  by some set $Q_y$ of $f + 1$ dependency service nodes. Any two sets of $f +
  1$ nodes must intersect ($f+1$ is a majority of $2f+1$). Consider a
  dependency service node $d$ in the intersection of $Q_x$ and $Q_y$. $d$
  received both $\msg{v_x, x}$ and $\msg{v_y, y}$. Without loss of generality,
  assume it received $\msg{v_y, y}$ second. Then, when $d$ received $\msg{v_y,
  y}$, $\msg{v_x, x}$ was already in its \textsf{commands}, so it must have
  included $v_x$ in its computed dependencies for $v_y$. $\deps{}_y$ is a union
  of dependencies that includes the dependencies computed by $d$. Thus, $v_x
  \in \deps{}_y$. This is illustrated in \figref{DepServiceProof}.
\end{proof}

{\tikzstyle{proc}=[draw, thick, circle]
\tikzstyle{quorum}=[draw, thick, rounded corners]
\tikzstyle{pointer}=[thick, -latex]

\begin{figure}[ht]
  \centering
  \begin{tikzpicture}
    \node[proc] (d0) at (0, 0) {$d_0$};
    \node[proc] (d1) at (1, 0) {$d_1$};
    \node[proc] (d2) at (2, 0) {$d_2$};
    \node (x) at (0.5, 0.8) {\textcolor{flatred}{$v_x$} quorum $Q_x$};
    \node (y) at (1.5, -0.8) {\textcolor{flatblue}{$v_y$} quorum $Q_y$};

    \draw[quorum, draw=flatred]
      ($(d0.north west) + (-0.2, 0.3)$) rectangle
      ($(d1.south east) + (0.2, -0.2)$);
    \draw[quorum, draw=flatblue]
      ($(d1.north west) + (-0.2, 0.2)$) rectangle
      ($(d2.south east) + (0.2, -0.3)$);

    \node[inner sep=1pt,
          anchor=south west] (description) at ($(d1) + (0.75, 0.75)$) {%
      $d_1$ receives $v_x$ and $v_y$
    };
    \draw[pointer] (description.south west) to (d1);
  \end{tikzpicture}
  \caption{An illustration of the proof of \thmref{DepServiceInvariant}.}%
  \figlabel{DepServiceProof}
\end{figure}
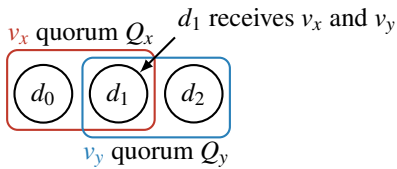}

Note that if the dependency service produces responses $\msg{v_x, x,
\deps{}_x}$ and $\msg{v_y, y, \deps{}_y}$ for conflicting commands $x$ and $y$,
it may include $v_x \in \deps{}_y$ \emph{and} $v_y \in \deps{}_x$. For example,
if dependency service node $d_1$ receives $x$ then $y$ while dependency service
node $d_2$ receives $y$ then $x$, then dependencies formed from $d_1$ and $d_2$
will have $v_x$ and $v_y$ in each other's dependencies. This is the reason why
BPaxos graphs may develop cycles.

Also note that the dependency service is an independent module within BPaxos.
The dependency service is unaware of consensus, or BPaxos graphs, or state
machines, or any other detail outside of the dependency service. The dependency
service can be completely understood in isolation. In contrast, dependency
computation in EPaxos and Caesar is tightly coupled with the rest of the
protocol. For example, in Caesar, every command is assigned a timestamp. If a
node receives two commands out of timestamp order, it must first wait to see if
the higher timestamp command gets chosen with a dependency on the lower
timstamp command before it is able to compute the lower timestamp command's
dependencies. This coupling prevents us from understanding dependency
computation in isolation.

\subsection{Leaders}
When a leader receives a command $x$ from a client, it assigns $x$ a globally
unique vertex id $v_x$. The mechanism by which leaders generate unique ids is
unimportant. You can use any mechanism you would like as long as ids are globally
unique. In our implementation, a vertex id is a tuple of the leader's index and
a monotonically increasing id beginning at $0$. For example, leader $2$
generates vertex ids $(2, 0), (2, 1), (2, 2)$, and so on.

After generating a vertex id $v_x$, the leader sends $\msg{v_x, x}$ to all
dependency service nodes, aggregates the dependencies from $f+1$ of them, and
forwards the dependencies to a proposer.

\subsection{Proposers and Consensus Service}
When a proposer receives a message $\msg{v_x, x, \deps{}_x}$, it proposes to
the consensus service that the value $(x, \deps{}_x)$ be chosen for vertex
$v_x$. The consensus service implements one instance of consensus for every
vertex, and eventually informs the proposer of the value $(x', \deps{}_x')$
that was chosen for vertex $v_x$. In the normal case, $(x', \deps{}_x')$ is
equal to $(x, \deps{}_x)$, but the consensus service is free to choose any
value proposed for vertex $v_x$.

You can implement the consensus service with any consensus protocol that you would
like. In our implementation of BPaxos, BPaxos proposers are Paxos proposers,
and the consensus service is implemented as $2f+1$ Paxos acceptors. We
implement Paxos with the standard optimization that phase 1 of the protocol can
be skipped in round $0$ (a.k.a.\ ballot $0$). Doing so, and partitioning vertex
ids uniformly across proposers, the proposers can get a value chosen in one
round trip to the acceptors (in the common case). This optimization is very
similar to the one done in MultiPaxos.

Again, note that the consensus service is an independent module that we can
understand in isolation. The consensus service implements consensus, and that
is it. It is unaware of dependencies, graphs, or any other detail of the
protocol. Moreover, note that the consensus service is not specialized at all
to BPaxos. We are able to use the Paxos protocol without modification. This
lets us avoid having to prove a specialized implementation of consensus
correct.

\subsection{Replicas}
Every BPaxos replica maintains a BPaxos graph and an instance of a state
machine. Every state machine begins in the same initial state. Upon receiving a
message $\msg{v_x, x, \deps{}_x}$ from a proposer, a replica adds vertex $v_x$
to its graph with command $x$ and with edges to $\deps{}_x$. As discussed
earlier, the replicas execute their graphs in reverse topological order, one
component at a time. When a replica is ready to execute a command $x$, it
passes it to the state machine. The state machine transitions to a new state
and produces some output $o$. One replica then returns $o$ to the client that
initially proposed $x$. In particular, given $n$ replicas, $r_i$ returns
outputs to clients for vertices $v_x$ where $\text{hash}(v_x) \% n = i$.

\subsection{Summary}
In summary, BPaxos is composed of five modules: leaders, dependency service
nodes, proposers, a consensus service, and replicas. Clients propose commands;
leaders assign unique ids to commands; the dependency service computes
dependencies (ensuring that conflicting commands depend on each other); the
proposers and consensus service reach consensus on every vertex; and replicas
execute commands.

\subsection{Fault Tolerance and Recovery}
BPaxos can tolerate up to $f$ failures. By deploying $f+1$ leaders, proposers,
and replicas, BPaxos guarantees that at least one of each is operational after
$f$ failures. The dependency service deploys $2f+1$ dependency service nodes,
ensuring that at a quorum of $f+1$ nodes is available despite $f$ failures. The
consensus service tolerates $f$ failures by assumption. In our implementation,
we use $2f+1$ Paxos acceptors, as is standard.

However, despite this, failures can still lead to liveness violations if we are
not careful. A replica executes vertex $v_x$ only after it has executed $v_x$'s
dependencies. If one of $v_x$'s dependencies has not yet been chosen, then the
execution of $v_x$ is delayed. For example, in \figref{ExampleBPaxosExecution},
the execution of $v_2$ is delayed until after $v_1$ has been chosen and
executed.

If a vertex $v_x$ depends on a vertex $v_y$ that remains forever unchosen, then
$v_x$ is never executed. This situation is rare, but possible in the event of
failures. For example, if two leaders $l_x$ and $l_y$ concurrently send
commands $x$ and $y$ in vertices $v_x$ and $v_y$ to the dependency service, and
if $l_y$ then crashes, it is possible that $v_x$ gets chosen with a dependency
on $v_y$, but $v_y$ remains forever unchosen.

Dealing with these sorts of failure scenarios to ensure that every command
eventually gets chosen is called \defword{recovery}. Every state machine
replication protocol has to implement some form of recovery, and for many
protocols (though not all protocols), recovery is its most complicated part.

\newcommand{\noop}{\text{noop}}
Fortunately, BPaxos' modularity leads to a very simple recovery protocol. When
a replica notices that a vertex $v_x$ has been blocked waiting for another
vertex $v_y$ for more than some configurable amount of time, the replica
contacts the consensus service and proposes that a no-operation command $\noop$
be chosen for vertex $v_y$ with no dependencies. $\noop$ is a special command
that does not affect the state machine and does not conflict with any other
command. Eventually, the consensus protocol returns the chosen value to the
replica, and the execution of $v_x$ can proceed.

}
{\section{Disaggregating and Scaling}
BPaxos' modular design leads to high throughput in two ways: disaggregation and
scaling.

\subsection{Identifying Bottlenecks}
The throughput of a protocol is determined by its bottleneck. Before we discuss
BPaxos' throughput, we discuss how to identify the bottleneck of a protocol.
Identifying a bottleneck with complete accuracy is hard. Protocol bottlenecks
are affected by many factors including CPU speeds, network bandwidth, message
sizes, workload characteristics, and so on. To make bottleneck analysis
tractable, we make a major simplifying assumption. The assumption is best
explained by way of an example.

{\input{figures/common.tex}

\tikzstyle{proc}=[draw, circle, thick]
\tikzstyle{proclabel}=[inner sep=0pt]
\tikzstyle{fakeproclabel}=[inner sep=0pt, white]
\tikzstyle{comm}=[-latex, thick]
\tikzstyle{commnum}=[font=\large, fill=flatyellowalt!50, inner sep=3pt,
                     rounded corners]
\newcommand{\acceptorcolor}{flatorange}

\begin{figure}[ht]
  \centering
  \begin{tikzpicture}[scale=0.75]
    \node[proc] (c) at (0, 2) {$c$};
    \node[proc,
          fill=\proposercolor!25,
          label={[commnum, label distance=0.3cm]90:\bm{$2N + 2$}}]
          (p0) at (2, 2) {$p_0$};
    \node[proc,
          fill=\proposercolor!25]
          (p1) at (4, 2) {$p_1$};
    \node[proc,
          fill=\acceptorcolor!25,
          label={[commnum]-90:\bm{$2$}}]
          (a0) at (0, 0) {$a_0$};
    \node[proc,
          fill=\acceptorcolor!25,
          label={[commnum]-90:\bm{$2$}}]
          (a1) at (2, 0) {$a_1$};
    \node[proc,
          fill=\acceptorcolor!25,
          label={[commnum]-90:\bm{$2$}}]
          (a2) at (4, 0) {$a_2$};

    \crown{(p0.north)++(0, -0.15)}{1}{0.5}

    \node[proclabel, anchor=east] at (-0.5, 2) {Client};
    \node[fakeproclabel, anchor=west] (ps) at (5, 2) {Proposers};
    \halffill{ps}{\proposercolor!25}
    \node[proclabel, anchor=west] (ps) at (5, 2) {Proposers};
    \node[fakeproclabel, anchor=west] (as) at (5, 0) {$N$ Acceptors};
    \halffill{as}{\acceptorcolor!25}
    \node[proclabel, anchor=west] (as) at (5, 0) {$N$ Acceptors};

    \draw[comm, bend left=10] (c) to (p0);
    \draw[comm, bend left=10] (p0) to (a0);
    \draw[comm, bend left=10] (p0) to (a1);
    \draw[comm, bend left=10] (p0) to (a2);
    \draw[comm, bend left=10] (a0) to (p0);
    \draw[comm, bend left=10] (a1) to (p0);
    \draw[comm, bend left=10] (a2) to (p0);
    \draw[comm, bend left=10] (p0) to (c);
  \end{tikzpicture}
  \caption{%
    MultiPaxos' throughput bottleneck. Every node is annotated with the number
    of messages that it sends and receives. MultiPaxos' throughput is
    proportional to $\frac{1}{2N+2}$.
  }
  \figlabel{MultiPaxosBottleneck}
\end{figure}
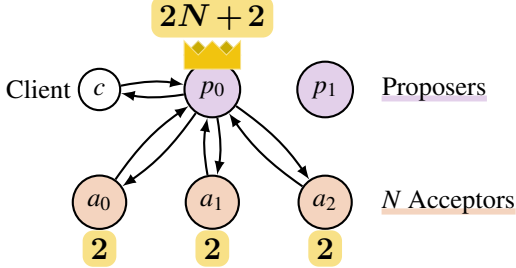}

Consider the execution of MultiPaxos shown in \figref{MultiPaxosBottleneck} in
which a client proposes a command $x$. The execution involves $N \geq 2f+1$
acceptors. We have annotated each node with the number of messages it sends and
receives in the process of handling $x$. The leader $p_0$ processes $2N+2$
messages, and every acceptor processes $2$ messages. Our major assumption is
that the time required for each node to process command $x$ is directly
proportional to the number of messages that it processes. Thus, the leader
takes time proportional to $2N+2$, and the acceptors take time proportional to
$2$. This means that the leader is the bottleneck, and the protocol's
throughput is directly proportional to $\frac{1}{2N+2}$, the inverse of the
time required by the bottleneck component.

While our assumption is simplistic, we will see in \secref{Evaluation} that
empirically it is accurate enough for us to identify the actual bottleneck of
protocols in practice. Now, we turn our attention to BPaxos. Consider the
execution of BPaxos shown in \figref{BPaxosBottleneck}. We have $N \geq 2f+1$
dependency service nodes, $N$ acceptors, $L \geq f+1$ leaders, $L$ proposers,
and $R \geq f+1$ replicas\footnote{We can have a different number of leaders
and proposers, but letting them be equal simplifies the example.}.

{\input{figures/common.tex}

\tikzstyle{proc}=[draw, circle, thick, inner sep=2pt]
\tikzstyle{leader}=[proc, fill=\leadercolor!25]
\tikzstyle{proposer}=[proc, fill=\proposercolor!25]
\tikzstyle{replica}=[proc, fill=\replicacolor!25]
\tikzstyle{client}=[proc, fill=\clientcolor!25]
\tikzstyle{proclabel}=[inner sep=0pt, darkgray]
\tikzstyle{fakeproclabel}=[inner sep=0pt, white]
\tikzstyle{comm}=[-latex, thick]
\tikzstyle{commnum}=[font=\large, fill=flatyellowalt!50, inner sep=3pt,
                     rounded corners]
\tikzstyle{service}=[draw, rounded corners, align=center, thick]
\tikzstyle{depnode}=[proc, fill=\depservicecolor!25]
\tikzstyle{acceptor}=[proc, fill=\consensuscolor!25]
\tikzstyle{module}=[draw, thick, flatgray, rounded corners]

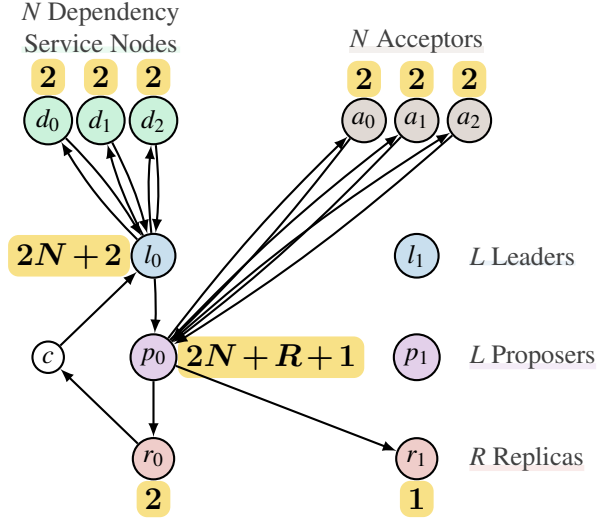
\begin{figure}[t]
  \centering

  \begin{tikzpicture}[yscale=0.9, xscale=0.70]
    \node[client] (c) at (0, 1.5) {$c$};
    \node[leader, label={[commnum]180:\bm{$2N+2$}}] (l0) at (2, 3) {$l_0$};
    \node[leader] (l1) at (7, 3) {$l_1$};
    \node[proposer, label={[commnum]0:\bm{$2N+R+1$}}] (p0) at (2, 1.5) {$p_0$};
    \node[proposer] (p1) at (7, 1.5) {$p_1$};
    \node[replica, label={[commnum]270:\bm{$2$}}] (r0) at (2, 0) {$r_0$};
    \node[replica, label={[commnum]270:\bm{$1$}}] (r1) at (7, 0) {$r_1$};
    \node[depnode, label={[commnum]90:\bm{$2$}}] (d0) at (0, 5) {$d_0$};
    \node[depnode, label={[commnum]90:\bm{$2$}}] (d1) at (1, 5) {$d_1$};
    \node[depnode, label={[commnum]90:\bm{$2$}}] (d2) at (2, 5) {$d_2$};
    \node[acceptor, label={[commnum]90:\bm{$2$}}] (a0) at (6, 5) {$a_0$};
    \node[acceptor, label={[commnum]90:\bm{$2$}}] (a1) at (7, 5) {$a_1$};
    \node[acceptor, label={[commnum]90:\bm{$2$}}] (a2) at (8, 5) {$a_2$};

    \draw[comm] (c) to (l0);
    \draw[comm, bend left=7] (l0) to (d0);
    \draw[comm, bend left=7] (l0) to (d1);
    \draw[comm, bend left=7] (l0) to (d2);
    \draw[comm, bend left=7] (d0) to (l0);
    \draw[comm, bend left=7] (d1) to (l0);
    \draw[comm, bend left=7] (d2) to (l0);
    \draw[comm, bend left=4] (l0) to (p0);
    \draw[comm, bend left=4] (p0) to (a0);
    \draw[comm, bend left=4] (p0) to (a1);
    \draw[comm, bend left=4] (p0) to (a2);
    \draw[comm, bend left=4] (a0) to (p0);
    \draw[comm, bend left=4] (a1) to (p0);
    \draw[comm, bend left=4] (a2) to (p0);
    \draw[comm] (p0) to (r0);
    \draw[comm] (p0) to (r1);
    \draw[comm] (r0) to (c);

    \node[draw, fakeproclabel, anchor=south, align=center] (depnodes) at (1, 6)
      {$N$ Dependency\\Service Nodes};
    \quarterfill{depnodes}{\depservicecolor!10}
    \node[proclabel, align=center] at (depnodes) {$N$ Dependency\\Service Nodes};

    \node[draw, fakeproclabel, anchor=south] (acceptors) at (7, 6) {$N$ Acceptors};
    \halffill{acceptors}{\consensuscolor!10}
    \node[proclabel] at (acceptors) {$N$ Acceptors};

    \node[fakeproclabel, anchor=west] (leaders) at (8, 3) {$L$ Leaders};
    \halffill{leaders}{\leadercolor!10}
    \node[proclabel] at (leaders) {$L$ Leaders};

    \node[fakeproclabel, anchor=west] (proposers) at (8, 1.5) {$L$ Proposers};
    \halffill{proposers}{\proposercolor!10}
    \node[proclabel] at (proposers) {$L$ Proposers};

    \node[fakeproclabel, anchor=west] (replicas) at (8, 0) {$R$ Replicas};
    \halffill{replicas}{\replicacolor!10}
    \node[proclabel] at (replicas) {$R$ Replicas};
  \end{tikzpicture}

  \caption{%
    BPaxos' throughput bottleneck. Every node is annotated with the number of
    messages that it sends and receives. BPaxos' throughput is proportional to
    $\frac{L}{2N+R+1}$.
  }
  \figlabel{BPaxosBottleneck}
\end{figure}}

Again, we annotate each node with the number of messages it processes to handle
the client's command. The dependency service nodes and acceptors process two
messages each. The replicas process either one or two messages---depending on
whether they are returning a response to the client---for an average of
$1+\frac{1}{R}$. The leaders and proposers process significantly more messages,
$2N+2$ and $2N+R+1$ messages respectively. Thus, the throughput through a
\emph{single} leader and proposer is proportional to $\frac{1}{2N+R+1}$. Unlike
MultiPaxos though, BPaxos does not have a single leader. All $L$ of the leaders
and proposers execute concurrently, with client commands divided amongst them.
With $L$ leaders and proposers, BPaxos' throughput is proportional to
$\frac{L}{2N+R+1}$.

\subsection{Disaggregation}
Many state machine replication protocols pack multiple logical nodes onto a
single physical node. We could do something similar. We could deploy $N=L=R$
dependency service nodes, acceptors, leaders, proposers, and replicas across
$N$ physical ``super nodes'', with one of each component co-located on a single
physical machine. This would reduce the latency of the protocol by two network
delays and open the door for optimizations that could reduce the latency even
further.

However, aggregating logical components together would worsen our bottleneck.
Now, for a given command, a super node would have to process the messages of a
dependency service node, an acceptor, a leader, a proposer, and a replica. With
the bottleneck component processing more messages per command, the throughput
of the protocol decreases. Disaggregating the components allows for pipeline
parallelism in which load is more evenly balanced across the components.

\subsection{Scaling}
Scaling is a classic systems technique that is used to increase the throughput
of a system. However, to date, consensus protocols have not been able to take
full advantage of scaling. Conventional wisdom for replication protocols
suggests that we use as few nodes as possible. Returning to
\figref{MultiPaxosBottleneck}, we see this conventional wisdom in action. The
throughput of MultiPaxos is proportional to $\frac{1}{2N+2}$. Adding more
proposers does not do anything, and adding more acceptors (i.e.\ increasing
$N$) \emph{lowers} the throughput.

BPaxos revises conventional wisdom and notes that while some components are
hard or impossible to scale (e.g., acceptors), other components scale
trivially. Serendipitously, the components that are easy to scale turn out to
be the same components that are a throughput bottleneck.

More specifically, we learned from \figref{BPaxosBottleneck} that BPaxos'
throughput is proportional to $\frac{L}{2N+R+1}$ with the $L$ leaders and
proposers being the bottleneck. To increase BPaxos' throughput, we simply
increase $L$. We can increase the number of leaders and proposers until they
are no longer the bottleneck. This pushes the bottleneck to either the
dependency service nodes, the acceptors, or the replicas. Fortunately, these
nodes only process at most two messages per command. This is equivalent to an
unreplicated state machine which must at least receive and execute a command
and reply with the result. Thus, we have effectively shrunk the throughput
bottleneck to its limit.

Note that we are able to perform this straightforward scaling because BPaxos'
components are modular. When we co-locate components together, $L=N=R$,
and it is impossible for us to increase $L$ (which increases throughput) without
increasing $N$ and $R$ (which decreases throughput). Modularity allows us
to scale each component independently.
}
{\section{Practical Considerations}\seclabel{PracticalConsiderations}

\subsection{Ensuring Exactly Once Semantics}
If a client proposes a command to a state machine replication protocol but
does not hear back quickly enough, it resends the command to the protocol to
make sure that the command eventually gets executed. Thus, a replication
protocol might \emph{receive} a command more than once, but it has to guarantee
that it never \emph{executes} the command more than once. Executing a command
more than once would violate exactly once semantics.

Non-generalized protocols like Paxos~\cite{van2015paxos}, Viewstamped
Replication~\cite{liskov2012viewstamped}, and Raft~\cite{ongaro2014search} all
employ the following technique to avoid executing a command more than once.
First, before a client proposes a command to a replication protocol, it
annotates the command with a monotonically increasing integer-valued id.
Moreover, clients only send one command at a time, waiting to receive a
response from one command before sending another. Second, every replica
maintains a \defword{client table}, like the one illustrated below. A client
table has one entry per client. The entry for a client records the largest id
of any command that the replica has executed for that client, along with the
result of executing the command with that id. A replica only executes commands
for a client if it has a larger id than the one recorded in the client table.
If it receives a command with the same id as the one in the client table, it
replies with the recorded output instead of executing the command a second
time.

\begin{center}
  \begin{tabular}{|c|c|c|}
    \hline
    \textbf{Client} & \textbf{Id} & \textbf{Output} \\\hline
    $10.31.14.41$   & $2$         & ``foo'' \\
    $10.54.13.123$  & $1$         & ``bar'' \\\hline
  \end{tabular}
\end{center}

Naively applying this same trick to BPaxos (or any generalized protocol) is
unsafe. For example, imagine a client issues command $x$ with id $1$. The
command gets chosen and is executed by replica 1. Then, the client issues
non-conflicting command $y$ with id $2$. The command gets chosen and is
executed by replica 2. Because $y$ has a larger id than $x$, replica 2 will
never execute $x$.

To fix this bug, a replica must record the ids of \emph{all} commands that it
has executed for a client, along with the output corresponding to the largest
of these ids. Replicas only execute commands they have not previously executed,
and relay the cached output if they receive a command with the corresponding
id.


\subsection{Dependency Compaction}
Upon receiving a command $x$ in vertex $v_x$, a dependency service node returns
the set of all previously received vertices with commands that conflict with
$x$. Over time, as the dependency service receives more and more commands,
these dependency sets get bigger and bigger. As the dependency sets get bigger,
BPaxos' throughput decreases because more time is spent sending these large
dependency sets, and less time is spent doing useful work.

To combat this, a BPaxos dependency service node has to compact dependencies in
some way. Recall that BPaxos leader $i$ creates vertex ids $(i, 0), (i, 1), (i,
2)$, and so on. Thus, vertex ids across all the leaders form a two-dimensional
array with one column for every leader index and one row for every
monotonically increasing id.

{\newlength{\depcompactioninnersep}
\setlength{\depcompactioninnersep}{4pt}
\newlength{\depcompactionlinewidth}
\setlength{\depcompactionlinewidth}{1pt}
\newlength{\depcompactionwidth}
\setlength{\depcompactionwidth}{\widthof{$X$}+2\depcompactioninnersep}
\newcommand{\logindexcolor}{flatred}

\tikzstyle{entry}=[draw,
                   inner sep=\depcompactioninnersep,
                   line width=\depcompactionlinewidth,
                   minimum height=\depcompactionwidth,
                   minimum width=\depcompactionwidth]
\tikzstyle{index}=[color=\logindexcolor,
                   inner sep=\depcompactioninnersep,
                   line width=\depcompactionlinewidth,
                   minimum height=\depcompactionwidth,
                   minimum width=\depcompactionwidth]
\tikzstyle{deps}=[fill=flatblue!20]
\tikzstyle{dep}=[thick, -latex]

\begin{figure}[ht]
  \centering
  \begin{tikzpicture}
    \node[entry, deps] (00) at (0, 0) {$b$};
    \node[entry, deps, right=-\depcompactionlinewidth of 00] (10) {$d$};
    \node[entry, deps, right=-\depcompactionlinewidth of 10] (20) {};
    \node[entry, deps, above=-\depcompactionlinewidth of 00] (01) {$a$};
    \node[entry, deps, right=-\depcompactionlinewidth of 01] (11) {};
    \node[entry, deps, right=-\depcompactionlinewidth of 11] (21) {$e$};
    \node[entry, above=-\depcompactionlinewidth of 01] (02) {};
    \node[entry, deps, right=-\depcompactionlinewidth of 02] (12) {$c$};
    \node[entry, right=-\depcompactionlinewidth of 12] (22) {};
    \node[entry, above=-\depcompactionlinewidth of 02] (03) {};
    \node[entry, right=-\depcompactionlinewidth of 03] (13) {};
    \node[entry, right=-\depcompactionlinewidth of 13] (23) {};

    \node[index, below=-\depcompactionlinewidth of 00] (c0) {$0$};
    \node[index, below=-\depcompactionlinewidth of 10] (c1) {$1$};
    \node[index, below=-\depcompactionlinewidth of 20] (c2) {$2$};
    \node[index, left=-\depcompactionlinewidth of 00] (r0) {$0$};
    \node[index, left=-\depcompactionlinewidth of 01] (r1) {$1$};
    \node[index, left=-\depcompactionlinewidth of 02] (r2) {$2$};
    \node[index, left=-\depcompactionlinewidth of 03] (r3) {$3$};

    \node at ($(c1) + (0, -\depcompactionwidth)$) {leader index};
    \node at ($(r1)!0.5!(r2) + (-\depcompactionwidth, 0)$) {id};
  \end{tikzpicture}
  \caption{An example of dependency compaction}
  \figlabel{DependencyCompaction}
\end{figure}
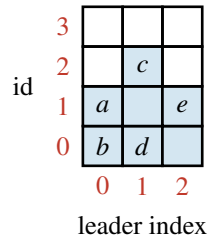
}

For example, consider a dependency service node that has received commands $a$,
$b$, $c$, $d$, and $e$ in vertices $(0, 1)$, $(0, 0)$, $(1, 2)$, $(1, 0)$, and
$(2, 1)$ as shown in \figref{DependencyCompaction}. \emph{Without dependency
compaction}, if the dependency service node receives a command that conflicts
with commands $a$, $b$, $c$, $d$, and $e$, it would return the vertex ids of
these five commands. In our example, the dependency service node returns only
five dependencies, but in a real deployment, the node could return hundreds of
thousands of dependencies.

\emph{With dependency compaction} on the other hand, the dependency service
node instead artificially adds more dependencies. In particular, for every
leader $i$, it computes the largest id $j$ for which a dependency $(i, j)$
exists. Then, it adds $\setst{(i, k)}{k \leq j}$ to the dependencies. In other
words, it finds the largest dependency in each column and then adds all of the
vertex ids below it as dependencies. In \figref{DependencyCompaction}, the
inflated set of dependencies is highlighted in blue. Even though more
dependencies have been added, the set of inflated dependencies can be
represented more compactly, with a single integer for every leader (i.e., the
id of the largest command for that leader). Thus, every BPaxos dependency set
can be succinctly represented with $N$ integers (for $N$ leaders).

}
{\section{Evaluation}\seclabel{Evaluation}

\subsection{Latency and Throughput}
{\begin{figure*}[ht]
  \centering
  \begin{subfigure}[c]{0.36\textwidth}
    \centering
    \includegraphics[width=\textwidth]{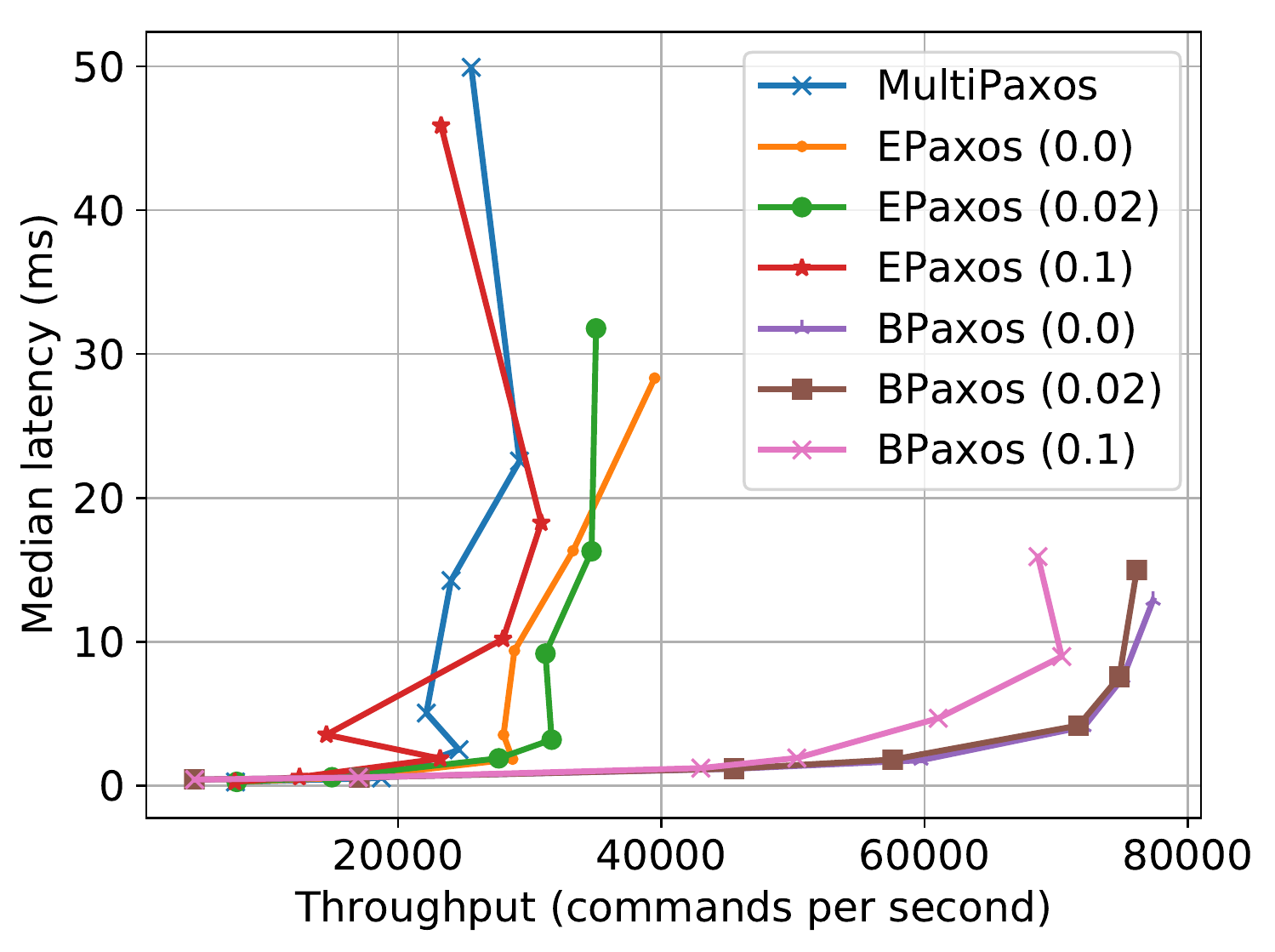}
    \caption{
      Latency-throughput curves for Multipaxos, EPaxos, and BPaxos. EPaxos and
      BPaxos are run with 0\%, 2\% and 10\% conflict rates. Here, $f = 1$.
    }\figlabel{EvalLtF1}
  \end{subfigure}
  \begin{subfigure}[c]{0.36\textwidth}
    \centering
    \includegraphics[width=\textwidth]{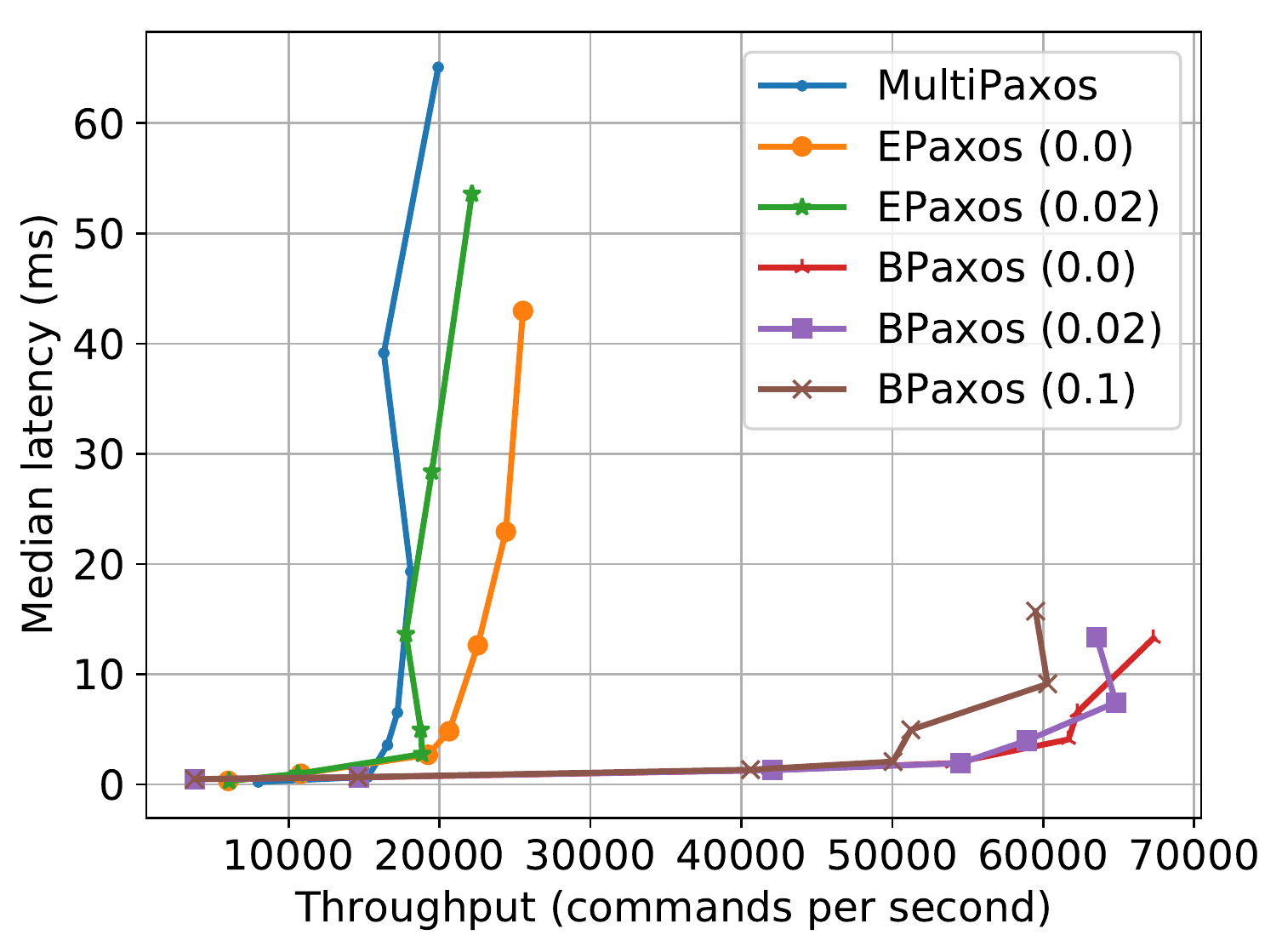}
    \caption{The same as \figref{EvalLtF1} but with $f=2$.}%
    \figlabel{EvalLtF2}
  \end{subfigure}
  \begin{subfigure}[c]{0.23\textwidth}
    \centering
    \small
    \begin{tabular}{lccc}
      \toprule
      \multicolumn{1}{c}{Protocol} &
      \multicolumn{3}{c}{Number of clients} \\
                    & 1    & 10   & 50 \\\midrule
      Multipaxos    & 0.24 & 0.52 & 2.49 \\
      EPaxos (0.0)  & 0.25 & 0.56 & 1.83 \\
      EPaxos (0.02) & 0.25 & 0.57 & 1.89 \\
      EPaxos (0.1)  & 0.25 & 0.58 & 1.87 \\
      BPaxos (0.0)  & 0.41 & 0.56 & 1.16 \\
      BPaxos (0.02) & 0.41 & 0.56 & 1.17 \\
      BPaxos (0.1)  & 0.41 & 0.55 & 1.21 \\
      \bottomrule
    \end{tabular}
    \caption{%
      Median latency values (ms) from \figref{EvalLtF1}.
    }\figlabel{EvalLtTable}
  \end{subfigure}
  \caption{%
    Latency and throughput of Multipaxos, EPaxos, and BPaxos for varying number
    of clients, conflict rates, and values of $f$. Data is shown for $1$, $10$,
    $50$, $100$, $300$, $600$, and $1200$ clients.
  }\figlabel{EvalLt}
\end{figure*}
}

\paragraph{Experiment Description.}
We implemented MultiPaxos, EPaxos\footnote{%
  Note that we implement \emph{Basic} EPaxos, the algorithm outlined
  in~\cite{moraru2013proof}. In general, Basic EPaxos has larger quorums and
  simpler recovery compared to the complete EPaxos protocol which is described
  in~\cite{moraru2013there}. For $f=1$ though, the performance of the two
  protocols is practically identical.
}, and BPaxos in Scala\footnote{%
  To mitigate the effects of JVM garbage collection on our experiments, we run
  our experiments with a large heap size of 32GB and run experiments for only a
  short amount of time.
}. Here, we measure the throughput and latency of the three protocols with
respect to three parameters: the number of clients, the conflict rate, and the
parameter $f$.

\begin{itemize}
  \item \textbf{Clients.}
    Clients propose commands in a closed loop. That is, after a client proposes
    a command, it waits to receive a response before proposing another command.
    We also run multiple clients in the same process, so deployments with a
    large number of clients (e.g., $1200$ clients) may use only a few client
    processes. We run $1$, $10$, $50$, $100$, $300$, $600$, and $1200$ clients.

  \item \textbf{Conflict rate.}
    The protocols replicate a key-value store state machine. Commands are
    single key gets or single key sets. With a conflict rate of $r$, $r$ of the
    commands are sets to a single key, while $(1 - r)$ of the commands are gets
    to other keys. Keys and values are both eight bytes. If commands are large,
    the data path and control path can be split, as in~\cite{biely2012s}. We
    run with $r=0$, $r=0.02$, and $r=0.1$. As described
    in~\cite{moraru2013there}, workloads in practice often have very low
    conflict rates.

  \item \textbf{$f$.}
    Recall that a protocol with parameter $f$ must tolerate at most $f$ failures.
    We run with $f=1$ and $f=2$.
\end{itemize}

We deploy the three protocols on m5.4xlarge EC2 instances within a single
availability zone. MultiPaxos deploys $f+1$ proposers and $2f+1$ acceptors.
EPaxos deploys $2f+1$ replicas. BPaxos deploys $2f+1$ dependency service nodes,
$2f+1$ acceptors, $f+1$ replicas, $5$ leaders and proposers when $f=1$, and
$10$ leaders and proposers when $f=2$. Every logical node is deployed on its
own physical machine, except that every BPaxos leader is co-located with a
BPaxos proposer. The protocols do not perform batching. All three protocols
implement thriftiness, a standard optimization~\cite{moraru2013there}.


\paragraph{Results.}
The benchmark results are shown in \figref{EvalLt}. In \figref{EvalLtF1} with
$f=1$, we see that MultiPaxos achieves a peak throughput of roughly 25,000 to
30,000 commands per second. EPaxos achieves a peak throughput of 30,000 to
40,000 depending on the conflict rate. BPaxos achieves 70,000 to 75,000,
nearly double that of EPaxos. Both EPaxos' and BPaxos' throughput decrease with
higher conflict rate. Higher conflict rates lead to graphs with more edges,
which increases the time required to topologically sort the graphs.

Note that the EPaxos implementation in~\cite{moraru2013there} achieves a peak
throughput of 45,000 to 50,000, slightly higher than our implementation. We
believe the discrepancy is due to implementation language (Go vs Scala) and
various optimizations performed in~\cite{moraru2013there} that we have not
implemented (e.g., a custom marshaling and RPC compiler~\cite{epaxos2019blog}).
We believe that if we apply the same optimizations to our implementations, all
three protocols' throughput would increase similarly.

In \figref{EvalLtF2}, with $f=2$, MultiPaxos' peak throughput has decreased to
20,000, EPaxos' peak throughput has decreased to 25,000, and BPaxos' peak
throughput has decreased to 65,000. As $f$ increases, the MultiPaxos leader has
to contact more nodes, so the drop in throughput is expected. With $f=2$,
EPaxos and BPaxos both have more leaders. More leaders increases the likelihood
of cycles, which slows the protocols down slightly. Moreover, when performing
dependency compaction as described in \secref{PracticalConsiderations}, the
number of dependencies scales with the number of leaders. BPaxos's peak
throughput is still roughly double that of EPaxos.

After sending a command, a BPaxos client must wait eight network delays to
receive a response. MultiPaxos and EPaxos require only four. Thus, under low
load, MultiPaxos and EPaxos have lower latency than BPaxos. In
\figref{EvalLtTable}, we see that with a single client, MultiPaxos and EPaxos
have a latency of roughly 0.25 ms, whereas BPaxos has a latency of 0.41. Under
high load though, BPaxos achieves lower latency. With 10 clients, the latency
of the three protocols is roughly even, and with 50 clients, BPaxos's latency
has already dropped below that of the other two protocols. In \figref{EvalLtF1}
and \figref{EvalLtF2}, we see that under higher loads of 600 and 1200 clients,
BPaxos's latency can be two to six times lower than the other two protocols.

Note that our results are specific to our deployment within a single data
center. With a geo-replicated deployment, MultiPaxos and EPaxos would both
outperform BPaxos. In this scenario, minimizing network delays is essential for
high performance.
Also note that BPaxos uses more machines than MultiPaxos and EPaxos in
order to achieve higher throughput via disaggregation and scaling. This makes
BPaxos a poor fit in resource constrained environments.

\subsection{Ablation Study}
{\begin{figure*}[ht]
  \centering
  \begin{subfigure}[b]{0.3\textwidth}
    \centering
    \includegraphics[width=\textwidth]{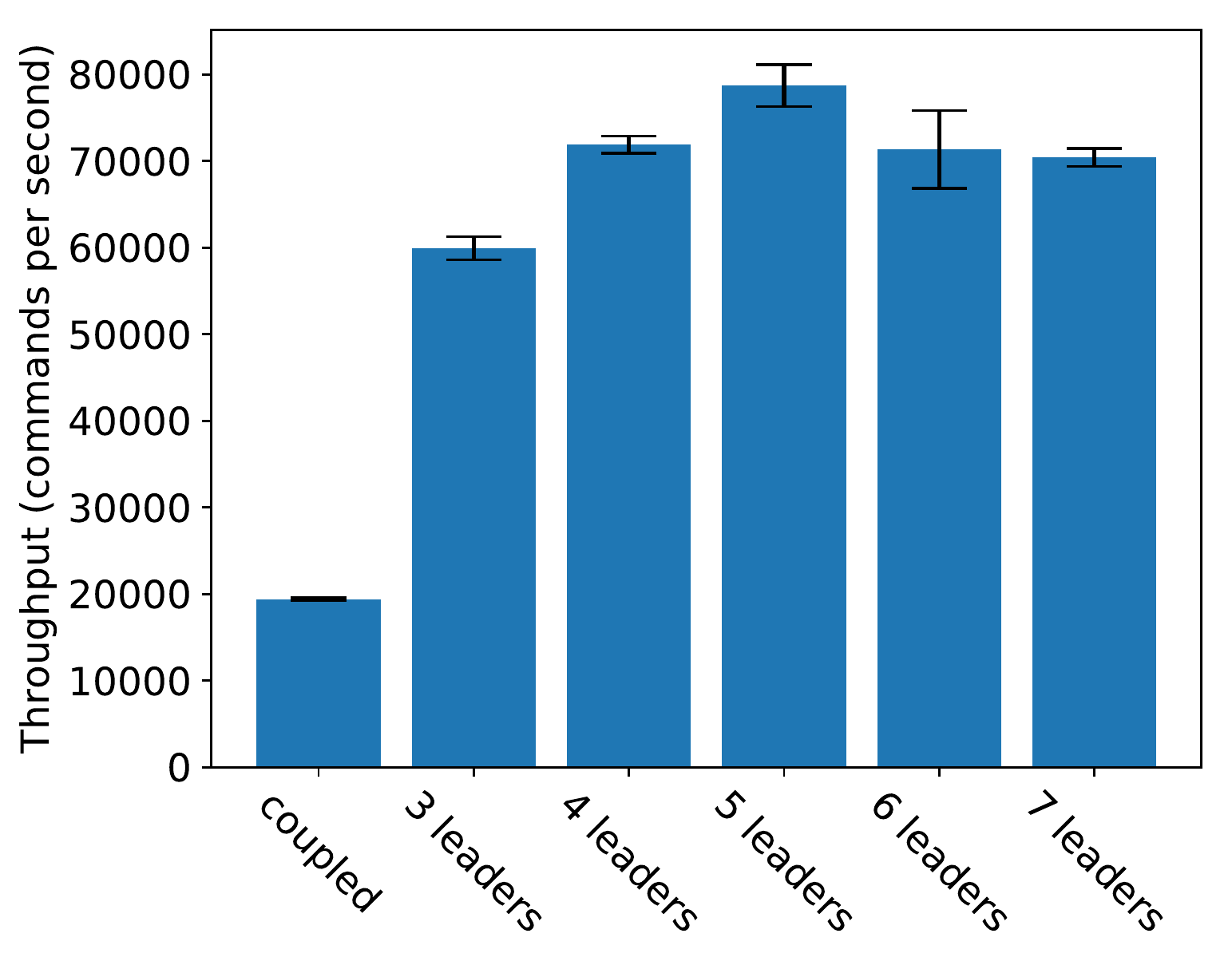}
    \caption{Throughput with 600 clients.}%
    \figlabel{EvalAblationHighLoadThroughput}
  \end{subfigure}\hspace{0.03\textwidth}
  \begin{subfigure}[b]{0.3\textwidth}
    \centering
    \includegraphics[width=\textwidth]{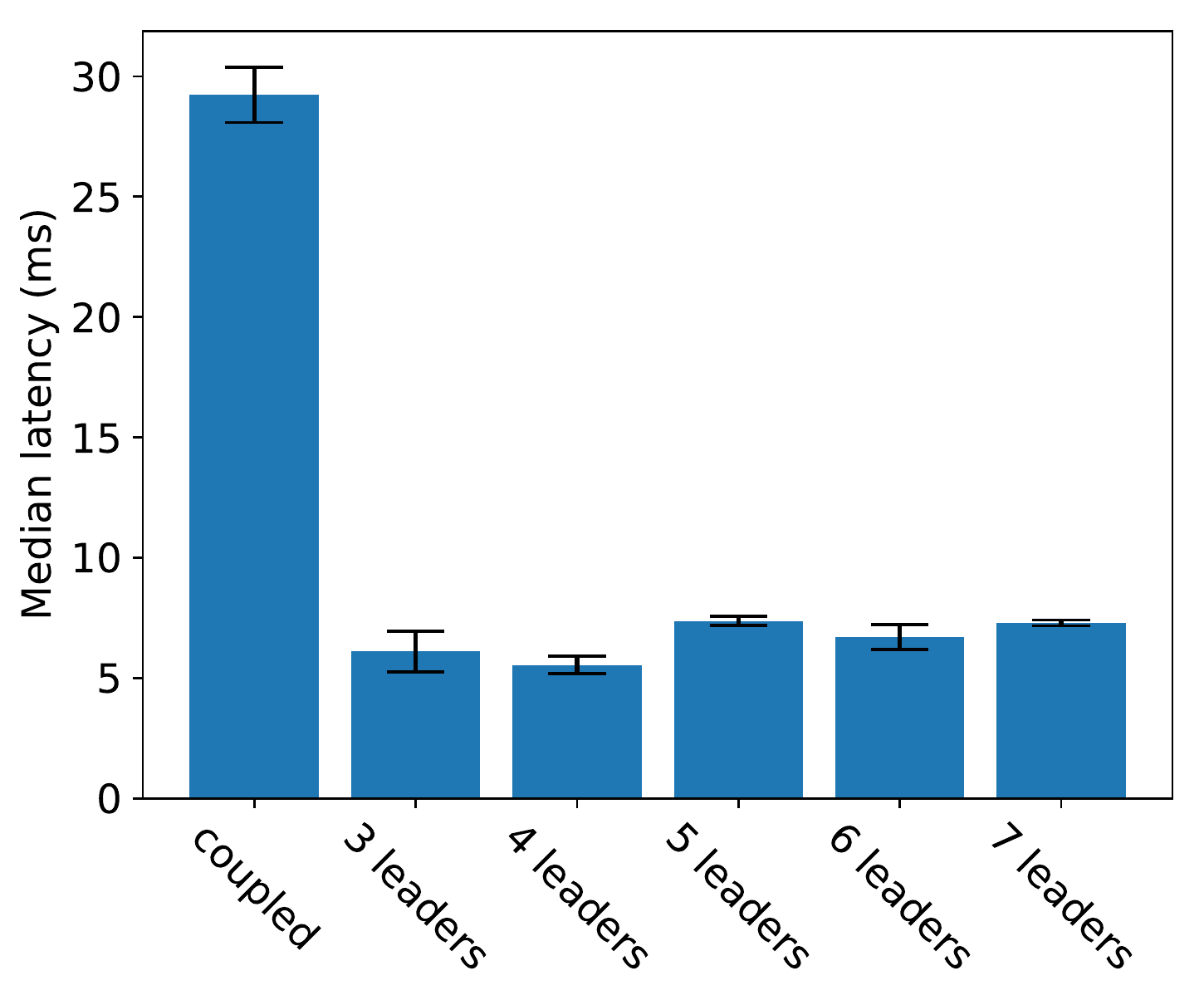}
    \caption{Median latency (ms) with 600 clients.}%
    \figlabel{EvalAblationHighLoadLatency}
  \end{subfigure}\hspace{0.03\textwidth}
  \begin{subfigure}[b]{0.3\textwidth}
    \centering
    \includegraphics[width=\textwidth]{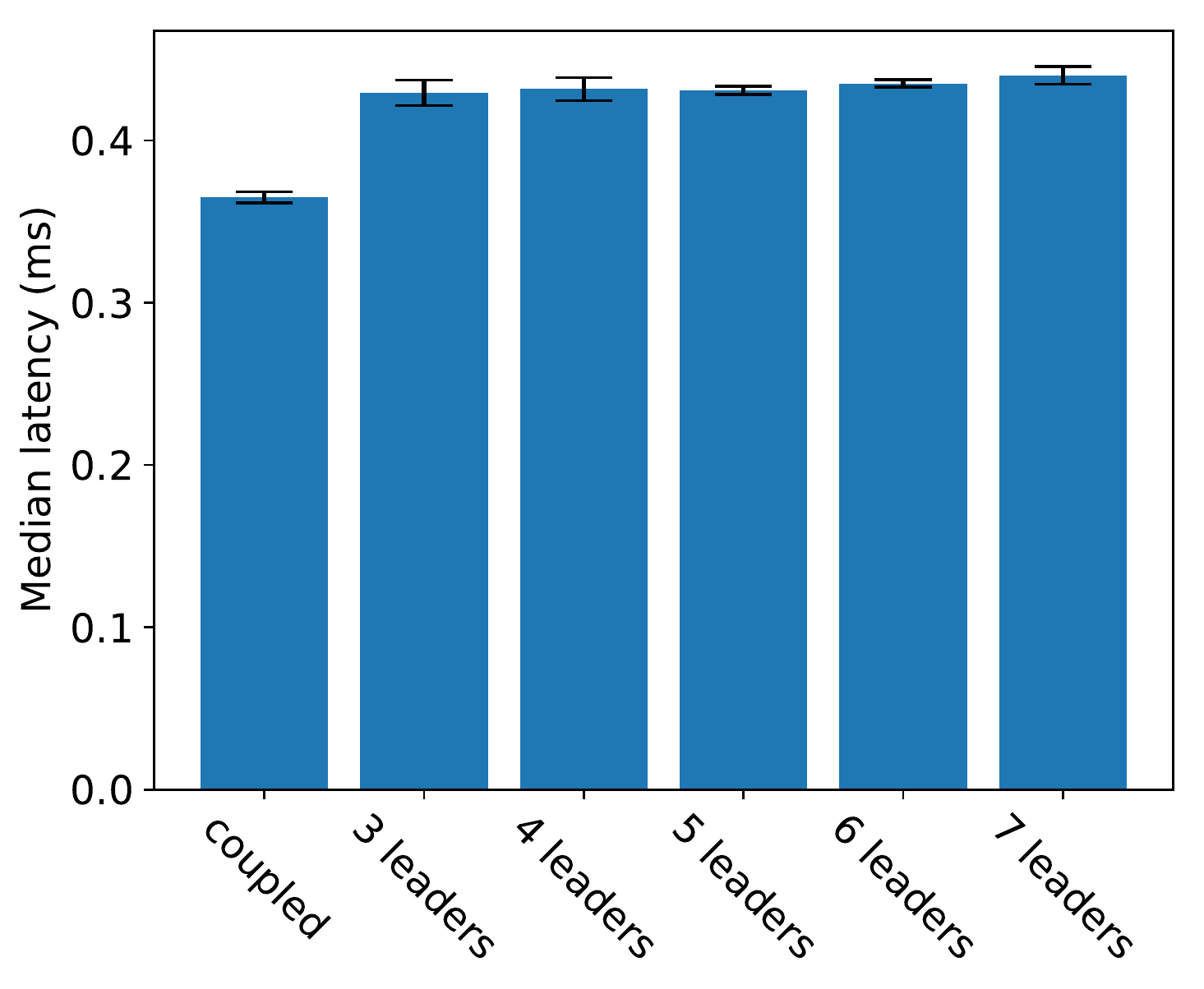}
    \caption{Median latency (ms) with one client.}%
    \figlabel{EvalAblationLowLoadLatency}
  \end{subfigure}
  \caption{%
    An ablation study showing the effect of disaggregation and scaling on
    throughput and latency with 600 clients and one client. Throughput for one
    client is not shown because it is simply the inverse of latency.
  }\figlabel{EvalAblation}
\end{figure*}
}

\paragraph{Experiment Description.}
The previous experiment showed that BPaxos can achieve roughly double the
throughput of EPaxos. Now, we analyze how BPaxos achieves these speedups. In
particular, we perform an ablation study to measure how BPaxos' disaggregation
and scaling affect its throughput. We repeat the experiment from above with
$f=1$, with $r=0$, and with $1$ and $600$ clients. We vary the number of
leaders from $3$ to $7$. Moreover, we also consider a ``coupled BPaxos''
deployment with three machines where each machine runs a single process that
acts as a leader, a dependency service node, a proposer, an acceptor, and a
replica. This artificially coupled BPaxos is similar to EPaxos in which every
replica plays many roles.


\paragraph{Results.}
The results of the experiment are shown in \figref{EvalAblation}. In
\figref{EvalAblationHighLoadThroughput}, we see the throughput of the coupled
BPaxos deployment is only 20,000 under high load. This is lower than both
MultiPaxos and EPaxos. When we decouple the protocol and run with three
leaders, the throughput increases threefold to 60,000. Disaggregating the nodes
introduces pipeline parallelism and reduces the load on the bottleneck
component. As we increase to five leaders, the throughput increases to a peak
of 75,000. At this point, the leaders are not the bottleneck and adding more
leaders only serves to slow down the protocol (for reasons similar to why the
$f=2$ deployment of BPaxos is slightly slower than the $f=1$ deployment).

In \figref{EvalAblationHighLoadLatency}, we see that the coupled protocol has
roughly six times the latency compared to the decoupled protocol under high
load. Moreover, the number of leaders doesn't have much of an impact on the
latency. In \figref{EvalAblationLowLoadLatency}, we see that the coupled
protocol has lower latency compared to the decoupled protocol under low load,
as fewer messages have to traverse the network. These results are consistent
with the previous experiment. Coupled protocols can achieve lower latency under
low load but decoupled protocols achieve higher throughput and lower latency
under high load.

In summary, both disaggregation and scaling contribute significantly to BPaxos'
increased throughput and lower latency under high load, and they also explain
why BPaxos has higher latency under low load.

%
%
\subsection{Batching}
Existing state machine replication protocols can perform batching to increase
their throughput at the cost of some latency~\cite{santos2012tuning,
santos2013optimizing, moraru2013proof}. BPaxos uses decoupling and scaling to
increase throughput at the cost of some latency. These two techniques
accomplish the same goal but are orthogonal. We can add batching to BPaxos to
increase its throughput even further. BPaxos leaders can collect batches of
commands from clients and place all of them within a single vertex. While
batching improves the throughput of all replication protocols, BPaxos' modular
design enables the protocol to take advantage of batching particularly well.

First, the overheads of receiving client messages and forming batches falls
onto the leaders. Because we can scale the leaders, these overheads can be
amortized until they are no longer a bottleneck. Moreover, the execution time
of proposers and acceptors increases linearly with the number of batches, not
the number of commands. Thus, increasing the batch size also amortizes their
overheads. Finally, as batch sizes grow, the number of vertices and edges in
the replicas' graphs shrinks. Thus, replicas can topologically sort the smaller
graphs faster.

We repeated the benchmarks from above with $f=1$ and $r=0$ with a batch size of
$1000$ and achieved a peak throughput of roughly $500,000$ commands per second
with a median latency of roughly 200 ms.
}
{\section{Related Work}

\paragraph{Paxos, VR, Raft}
MultiPaxos~\cite{lamport1998part, lamport2001paxos, van2015paxos,
lampson2001abcd, mazieres2007paxos}, Raft~\cite{mazieres2007paxos}, and
Viewstamped Replication~\cite{liskov2012viewstamped} are all single leader,
non-generalized state machine replication protocols. BPaxos has higher
throughput than these protocols because it is not bottlenecked by a single
leader. These protocols, however, have lower latency than BPaxos under low load
and are much simpler.

\paragraph{Mencius}
Mencius~\cite{mao2008mencius} is a multi-leader, non-generalized protocol in
which MultiPaxos log entries are round-robin partitioned among a set of
leaders. Because Mencius is not generalized, a log entry cannot be executed
until \emph{all} previous log entries have been executed. To ensure log entries
are being filled in appropriately, Mencius leaders perform all-to-all
communication between each other. This prevents leaders from scaling and
prevents other throughput-improving optimizations such as thriftiness.

\paragraph{Generalized GPaxos}
Generalized Paxos~\cite{lamport2005generalized} and GPaxos~\cite{sutra2011fast}
are generalized, but not fully multi-leader. Clients can send commands directly
to acceptors, behaving very much like a leader. However, in the face of
collisions, Generalized Paxos and GPaxos rely on a single leader to resolve the
collision. This single leader becomes a bottleneck in high contention workloads
and prevents scaling.

\paragraph{EPaxos and Caesar}
EPaxos~\cite{moraru2013there, moraru2013proof}, like BPaxos, is generalized and
multi-leader. EPaxos has lower latency than BPaxos (four network delays as
opposed to eight). EPaxos is a tightly coupled protocol. Every node acts as a
leader, dependency service node, proposer, acceptor, and replica. This
increases the load on the bottleneck nodes and also prevents disaggregation and
scaling. EPaxos, like Fast Paxos, optimistically takes a ``fast path'' before
sometimes reverting to a ``slow path''. This allows the protocol to execute a
command in four network delays in the best case, but fast paths significantly
complicate the protocol. For example, recovery in the face of fast paths can
deadlock if not implemented correctly. Caesar~\cite{arun2017speeding} is very
similar to EPaxos, with slight tweaks that increase the odds of the fast path
being taken.

\paragraph{A Family of Leaderless Generalized Algorithms}
In~\cite{losa2016brief}, Losa et al.\ present a generic architecture for
leaderless (what we call multi-leader) generalized consensus protocols. The
generic algorithm is very similar to BPaxos. In fact, some parts like the
dependency service are practically identical. However, the three page paper
does not present any implementations and focuses more on the theory behind
abstracting the commonalities shared by existing leaderless generalized
algorithms. BPaxos fleshes out the design and improves on the work by
discussing disaggregation, scaling, and practical considerations like ensuring
exactly once semantics and dependency compaction.

\paragraph{Multi-Core Paxos}
In~\cite{santos2013achieving}, Santos et al.\ describe how to increase the
throughput of a single MultiPaxos node by decomposing the node into multiple
components, with each component run on a separate core (e.g., one core for
sending messages, one for receiving messages, and so on). This work complements
BPaxos nicely. Santos et al.\ perform fine-grained decoupling to improve the
throughput of a single node, and BPaxos performs higher-level protocol
decoupling to improve the throughput of the entire protocol.

\paragraph{SpecPaxos, NOPaxos, CURP}
SpecPaxos~\cite{ports2015designing}, NOPaxos~\cite{li2016just}, and
CURP~\cite{park2019exploiting} all perform speculative execution to reduce
latencies as low as two network delays. However, speculative execution on the
fast path significantly increases the complexity of the protocols, and none of
the protocols focus on disaggregation or scaling as a means to increase
throughput.
}

\bibliographystyle{plain}
\bibliography{references}

\appendix
{\onecolumn
\section{BPaxos TLA+ Specification}\applabel{TlaSpec}
\begin{verbatim}
------------------------------ MODULE SimpleBPaxos -----------------------------
(******************************************************************************)
(* This is a specification of Simple BPaxos. To keep things simple and to     *)
(* make models more easily checkable, we abstract a way a lot of the          *)
(* unimportant details of Simple BPaxos. In particular, the specification     *)
(* does not model messages being sent between components and does not         *)
(* include leaders, proposers, or replicas. The consensus service is also     *)
(* left abstract. The core of Simple BPaxos is that dependency service        *)
(* responses (noops) are proposed to a consensus service. This core of the    *)
(* algorithm is what is modelled.                                             *)
(*                                                                            *)
(* Run `tlc SimpleBPaxosModel` to check the model.                            *)
(******************************************************************************)

EXTENDS Dict, Integers, FiniteSets

(******************************************************************************)
(* Constants                                                                  *)
(******************************************************************************)

\* The set of commands that can be proposed to BPaxos. In this specification,
\* every command can be proposed at most once. This is mostly to keep behaviors
\* finite. In a real execution of Simple BPaxos, a command can be proposed an
\* infinite number of times.
CONSTANT Command
ASSUME IsFiniteSet(Command)

\* The command conflict relation. Conflict is a symmetric relation over Command
\* such that two commands a and b conflict if (a, b) is in Conflict.
CONSTANT Conflict
ASSUME
    /\ Conflict \subseteq Command \X Command
    /\ \A ab \in Conflict : <<ab[2], ab[1]>> \in Conflict

\* We assume the existence of a special noop command that does not conflict
\* with any other command. Because noop is not in Command, it does not appear
\* in Conflict.
CONSTANT noop
ASSUME noop \notin Command

\* The set of dependency service nodes.
CONSTANT DepServiceNode
ASSUME IsFiniteSet(DepServiceNode)

\* The set of dependency service quorums. Every two quorums must interesct.
\* Typically, we deploy 2f + 1 dependency service replicas and let quorums be
\* sets of replicas of size f + 1.
CONSTANT DepServiceQuorum
ASSUME
    /\ \A Q \in DepServiceQuorum : Q \subseteq DepServiceNode
    /\ \A Q1, Q2 \in DepServiceQuorum : Q1 \intersect Q2 /= {}

--------------------------------------------------------------------------------

(******************************************************************************)
(* Variables and definitions.                                                 *)
(******************************************************************************)
\* In Simple BPaxos, vertex ids are of the form Q.i where Q is a leader and i is
\* a monotonically increasing id (intially zero). In this specification, we
\* don't even model Simple BPaxos nodes. So, we let instances be simple
\* integers. You might imagine we would say `VertexId == Nat`, but keeping
\* things finite helps TLC. Every command can be proposed at most once, so
\* allowing instances to range between 0 and |Command| works great.
VertexId == 0..Cardinality(Command)

\* A proposal is a command (or noop) and its dependencies.
Proposal == [cmd: Command \union {noop}, deps: SUBSET VertexId]

\* The proposal associated with noop. Noop doesn't conflict with any other
\* command, so its dependencies are always empty.
noopProposal == [cmd |-> noop, deps |-> {}]

\* A dependency graph is a directed graph where each vertex is labelled with an
\* vertex id and contains a command. We model the graph as a dictionary mapping
\* a vertex id to its command and dependencies.
DependencyGraph == Dict(VertexId, Proposal)

\* dependencyGraphs[d] is the dependency graph maintained on dependency
\* service node d.
VARIABLE dependencyGraphs

\* The next vertex id to assign to a proposed command. It is initially 0 and
\* incremented after every proposed command.
VARIABLE nextVertexId

\* A dictionary mapping vertex id to the command proposed with that vertex id.
VARIABLE proposedCommands

\* A dictionary mapping vertex id to the set of proposals proposed to the
\* consensus service in that instance.
VARIABLE proposals

\* A dictionary mapping vertex id to the proposal that was chosen by the
\* consensus service for that vertex id.
VARIABLE chosen

vars == <<
  dependencyGraphs,
  nextVertexId,
  proposedCommands,
  proposals,
  chosen
>>

TypeOk ==
  /\ dependencyGraphs \in Dict(DepServiceNode, DependencyGraph)
  /\ nextVertexId \in VertexId
  /\ proposedCommands \in Dict(VertexId, Command)
  /\ proposals \in Dict(VertexId, SUBSET Proposal)
  /\ chosen \in Dict(VertexId, Proposal)

--------------------------------------------------------------------------------

(******************************************************************************)
(* Actions.                                                                   *)
(******************************************************************************)

\* Propose a command `cmd` to Simple BPaxos. In a real implementation of Simple
\* BPaxos, a client would send the command to a leader, and the leader would
\* forward the command to the set of dependency service nodes. Here, we bypass
\* all that. The only thing to do here is to assign the command an instance and
\* make sure it hasn't already been proposed.
ProposeCommand(cmd) ==
  /\ cmd \notin Values(proposedCommands)
  /\ proposedCommands' = [proposedCommands EXCEPT ![nextVertexId] = cmd]
  /\ nextVertexId' = nextVertexId + 1
  /\ UNCHANGED <<dependencyGraphs, proposals, chosen>>

\* Given a dependency graph G and command cmd, return the set of vertices in G
\* that contain commands that conflict with cmd. For example, consider the
\* following dependency graph with commands b, c, and d in vertices v_b, v_c,
\* and v_d. If command a conflicts with c and d, then the dependencies of a are
\* v_c and v_d.
\*
\*                                 v_b     v_c
\*                                +---+   +---+
\*                                | b +---> c |
\*                                +-+-+   +---+
\*                                  |
\*                                +-v-+
\*                                | d |
\*                                +---+
\*                                 v_d
Dependencies(G, cmd) ==
  {v \in VertexId : G[v] /= NULL /\ <<cmd, G[v].cmd>> \in Conflict}

\* Here, dependency service node d processes a request in vertex v. Namely,
\* it adds v to its dependency graph (along with the command in
\* proposedCommands). Dependency service nodes also do not process a command
\* more than once. In a real Simple BPaxos implementation, the dependency
\* service node would receive a message from a leader and send dependencies
\* back to the leader. Also, a dependency service node could receive a request
\* from the leader more than once. We abstract all of this away.
DepServiceProcess(d, v) ==
  LET G == dependencyGraphs[d] IN
  /\ proposedCommands[v] /= NULL
  /\ G[v] = NULL
  /\ LET cmd == proposedCommands[v] IN
    /\ dependencyGraphs' = [dependencyGraphs EXCEPT ![d][v] =
                              [cmd |-> cmd, deps |-> Dependencies(G, cmd)]]
    /\ UNCHANGED <<nextVertexId, proposedCommands, proposals, chosen>>

\* Evalutes to whether a quorum of dependency service nodes have processed the
\* command in vertex v.
ExistsQuorumReply(Q, v) ==
  \A d \in Q : dependencyGraphs[d][v] /= NULL

\* Evaluates to the dependency service reply for vertex v from quorum Q of
\* dependency service nodes.
QuorumReply(Q, v) ==
  LET responses == {dependencyGraphs[d][v] : d \in Q} IN
  [cmd |-> (CHOOSE response \in responses : TRUE).cmd,
   deps |-> UNION {response.deps : response \in responses}]

\* Propose a noop gadget in vertex v to the consensus service. In a real
\* Simple BPaxos implementation, a proposer would propose a noop only
\* in some circumstances. In this model, we allow noops to be proposed at any
\* time.
ConsensusProposeNoop(v) ==
  /\ proposals' = [proposals EXCEPT ![v] = @ \union {noopProposal}]
  /\ UNCHANGED <<dependencyGraphs, nextVertexId, proposedCommands, chosen>>

\* Propose a dependency service reply in vertex v to the consensus service.
ConsensusPropose(v) ==
  \E Q \in DepServiceQuorum :
    /\ ExistsQuorumReply(Q, v)
    /\ proposals' = [proposals EXCEPT ![v] = @ \union {QuorumReply(Q, v)}]
    /\ UNCHANGED <<dependencyGraphs, nextVertexId, proposedCommands, chosen>>

\* Choose a value for vertex v.
ConsensusChoose(v) ==
  /\ proposals[v] /= {}
  /\ chosen[v] = NULL
  /\ chosen' = [chosen EXCEPT ![v] = CHOOSE g \in proposals[v] : TRUE]
  /\ UNCHANGED <<dependencyGraphs, nextVertexId, proposedCommands, proposals>>

--------------------------------------------------------------------------------

(******************************************************************************)
(* Specification.                                                             *)
(******************************************************************************)
Init ==
  /\ dependencyGraphs = [d \in DepServiceNode |-> [v \in VertexId |-> NULL]]
  /\ nextVertexId = 0
  /\ proposedCommands = [v \in VertexId |-> NULL]
  /\ proposals = [v \in VertexId |-> {}]
  /\ chosen = [v \in VertexId |-> NULL]

Next ==
  \/ \E cmd \in Command : ProposeCommand(cmd)
  \/ \E d \in DepServiceNode : \E v \in VertexId : DepServiceProcess(d, v)
  \/ \E v \in VertexId : ConsensusProposeNoop(v)
  \/ \E v \in VertexId : ConsensusPropose(v)
  \/ \E v \in VertexId : ConsensusChoose(v)

Spec == Init /\ [][Next]_vars

FairSpec == Spec /\ WF_vars(Next)

--------------------------------------------------------------------------------

(******************************************************************************)
(* Properties and Invariants.                                                 *)
(******************************************************************************)
\* The consensus service can choose at most command in any given instance.
ConsensusConsistency ==
  \A v \in VertexId :
    chosen[v] /= NULL => chosen'[v] = chosen[v]

AlwaysConsensusConsistency ==
  [][ConsensusConsistency]_vars

\* If two conflicting commands a and b yield dependencies deps(a) and deps(b)
\* from the dependency service, then a is in deps(b), or b is in deps(a), or
\* both.
DepServiceConflicts ==
  \A v1, v2 \in VertexId :
  \A Q1, Q2 \in DepServiceQuorum :
  IF v1 /= v2 /\ ExistsQuorumReply(Q1, v1) /\ ExistsQuorumReply(Q2, v2) THEN
     LET proposal1 == QuorumReply(Q1, v1)
         proposal2 == QuorumReply(Q2, v2) IN
     <<proposal1.cmd, proposal2.cmd>> \in Conflict =>
       v1 \in proposal2.deps \/ v2 \in proposal1.deps
  ELSE
    TRUE

\* Simple BPaxos should only choose proposed commands. This is inspired by [1].
\*
\* [1]: github.com/efficient/epaxos/blob/master/tla+/EgalitarianPaxos.tla
Nontriviality ==
  \A v \in VertexId :
    chosen[v] /= NULL =>
      \/ chosen[v].cmd \in Values(proposedCommands)
      \/ chosen[v].cmd = noop

\* If two conflicting commands a and b are chosen, then a is in deps(b), or b
\* is in deps(a), or both.
ChosenConflicts ==
  \A v1, v2 \in VertexId :
  IF v1 /= v2 /\ chosen[v1] /= NULL /\ chosen[v2] /= NULL THEN
    LET proposal1 == chosen[v1]
        proposal2 == chosen[v2] IN
     <<proposal1.cmd, proposal2.cmd>> \in Conflict =>
       v1 \in proposal2.deps \/ v2 \in proposal1.deps
  ELSE
    TRUE

\* True if every command is chosen.
EverythingChosen ==
  \A cmd \in Command :
    \E v \in VertexId :
      /\ chosen[v] /= NULL
      /\ chosen[v] = cmd

\* Fairness free theorem.
THEOREM
  Spec => /\ AlwaysConsensusConsistency
          /\ []DepServiceConflicts
          /\ []Nontriviality
          /\ []ChosenConflicts

\* True if no noops are chosen.
NoNoop ==
  ~ \E v \in VertexId :
    /\ chosen[v] /= NULL
    /\ chosen[v].cmd = noop

\* If no noops are chosen, then every command is chosen. This property is only
\* true for FairSpec.
NoNoopEverythingChosen ==
  []NoNoop => <>EverythingChosen

\* Fairness theorem.
THEOREM
  FairSpec => /\ AlwaysConsensusConsistency
              /\ []DepServiceConflicts
              /\ []Nontriviality
              /\ []ChosenConflicts
              /\ NoNoopEverythingChosen

================================================================================

--------------------------------- MODULE Dict ----------------------------------

(******************************************************************************)
(* TLA+ has the notion of functions. For example [A -> B] is the set of all   *)
(* functions from the set A to the set B. Functions are a lot like the        *)
(* dictionaries you find in a language like Python, except for one notable    *)
(* distinction. A function f \in [A \to B] is total, so every value a \in A   *)
(* must map to some value b \in B by way of f. Dictionaries from A to B, on   *)
(* the other hand, do not have to map every a \in A to some corresponding b   *)
(* \in B. This module builds up dictionaries out of functions. Doing so is    *)
(* relatively straightforward. We introduce a NULL value and model a          *)
(* Dictionary as a function [A \to B \cup {NULL}].                            *)
(******************************************************************************)

CONSTANT NULL

Dict(K, V) == [K -> V \cup {NULL}]

Keys(dict) == {k \in DOMAIN dict : dict[k] /= NULL}

Values(dict) == {dict[k] : k \in Keys(dict)}

Items(dict) == {<<k, dict[k]>> : k \in Keys(dict)}

================================================================================
\end{verbatim}
}
\end{document}